\documentclass{scrartcl}
\KOMAoptions{paper=letter}

\usepackage[todo]{group}

\usepackage{thm-restate}
\usepackage{bigfoot}

\newtheorem{innercustomthm}{Lemma}
\newenvironment{customlemma}[1]
  {\renewcommand\theinnercustomthm{#1}\innercustomthm}
  {\endinnercustomthm}

\title{Incentives in Social Decision Schemes with Pairwise Comparison Preferences}

\author{Felix Brandt, Technical University of Munich \and Patrick Lederer\thanks{Corresponding author. Email address: ledererp@in.tum.de}, Technical University of Munich \and Warut Suksompong, National University of Singapore}

\newenvironment{profile}{
\smallskip
\small
\setlength{\tabcolsep}{1em}
\centering 
\medmuskip=0mu\relax
	\thickmuskip=1mu\relax
\noindent
\begin{tabular}{rlllll}}
{\end{tabular}\\[1ex]}

\newenvironment{profilenarrow}{
\smallskip
\small
\setlength{\tabcolsep}{0.5em}
\centering 
\medmuskip=0mu\relax
	\thickmuskip=1mu\relax
\noindent
\begin{tabular}{rlllll}}
{\end{tabular}\\[1ex]}
	
\newcommand{\seti}[2]{i\hspace{3pt}\in\hspace{3pt} \medmuskip=0mu\relax
	\thickmuskip=1mu\relax [#1\dots#2]}
\newcommand{\set}[2]{\medmuskip=0mu\relax
	\thickmuskip=1mu\relax [#1\dots#2]}

\usepackage{comment}
\excludecomment{proofsketch}

\newcommand{\X}{\ensuremath{\mathit{\mathcal{X}}}\xspace}

\allowdisplaybreaks

\begin{document}

\maketitle

\begin{abstract}
    Social decision schemes (SDSs) map the ordinal preferences of individual voters over multiple alternatives to a probability distribution over the alternatives. In order to study the axiomatic properties of SDSs, we
    lift preferences over alternatives to preferences over lotteries using the natural---but little understood---pairwise comparison (\pc) preference extension. This extension postulates that one lottery is preferred to another if the former is more likely to return a preferred outcome. We settle three open questions raised by \citet{Bran17a}: \emph{(i)} there is no Condorcet-consistent SDS that satisfies \pc-strategyproofness; \emph{(ii)} there is no anonymous and neutral SDS that satisfies \pc-efficiency and \pc-strategyproofness; and \emph{(iii)} there is no anonymous and neutral SDS that satisfies \pc-efficiency and strict \pc-participation. All three impossibilities require $m\geq 4$ alternatives and turn into possibilities when $m\leq 3$. We furthermore settle an open problem raised by \citet{ABB14b} by showing that no path of \pc-improvements originating from an inefficient lottery may lead to a \pc-efficient lottery.
\end{abstract}


\textbf{JEL Classifications Code:} D7

\section{Introduction}

Incentives constitute a central aspect when designing mechanisms for multiple agents: mechanisms should incentivize agents to participate and to act truthfully. However, for many applications, guaranteeing these properties---usually called \emph{participation} and \emph{strategyproofness}---is a notoriously difficult task. This is particularly true for collective decision making, which studies the aggregation of preferences of multiple voters into a group decision, because strong impossibility theorems show that these axioms are in variance with other elementary properties \citep[see, e.g.,][]{Gibb73a,Satt75a,Moul88b}. For instance, the Gibbard-Satterthwaite theorem shows that every strategyproof voting rule is either dictatorial or imposing, and Moulin's No-Show paradox demonstrates that all Condorcet-consistent voting rules violate participation. 

A natural escape route in light of these negative results is to allow for randomization in the output of the voting rule. Rather than returning a single winner, a \emph{social decision scheme (SDS)} selects a lottery over the alternatives and the winner is eventually drawn at random according to the given probabilities. In order to study properties such as strategyproofness and participation as well as economic efficiency for SDSs, we need to make assumptions on how voters compare lotteries. The standard approach for this problem is to lift the voters' preferences over alternatives to preferences over lotteries by using the notion of \emph{stochastic dominance (\sd)}: a voter prefers a lottery to another one if the expected utility of the former exceeds that of the latter for \emph{every} utility representation consistent with his preferences over the alternatives \citep[see, e.g.,][]{Gibb77a,BoMo01a,BBEG16a}. 

Unfortunately, the negative results from deterministic social choice largely prevail when analyzing SDSs based on \sd preferences. For instance, \citet{Gibb77a} has shown that the only SDS that satisfies \sd-strategyproofness, non-imposition, and anonymity is the \emph{uniform random dictatorship (\rd)}, which chooses a voter uniformly at random and returns his favorite alternative. 
Similarly, \citet{BGP16c} have proven that Moulin's No-show paradox remains intact when defining participation based on \sd preferences. 
Independently of these negative results, the representation of preferences over lotteries via expected utility functions has come under scrutiny in decision theory \citep[e.g.,][]{Alla53a,KaTv79a,Mach89a,Anan09a}. 

As an alternative to traditional expected utility representations, some authors have proposed to postulate that one lottery is preferred to another if the former is more likely to return a preferred outcome \citep[][]{Blyt72a,Pack82a,Blav06a}.
The resulting preference extension is known as \emph{pairwise comparison (\pc)} and represents a special case of Fishburn's skew-symmetric bilinear utility functions \citep{Fish82c}. \citet{BBH15c} have shown that the No-Show paradox can be circumvented using \pc preferences. Moreover, \citet{BrBr17a} proved that \pc preferences constitute the \emph{only} domain of preferences within a rather broad class of preferences over lotteries (including all expected utility representations) that allows for preference aggregation that satisfies independence of irrelevant alternatives and efficiency, thus avoiding Arrow's impossibility. In both cases, the resulting SDS is the set of \emph{maximal lotteries (\ml}), which was proposed by \citet{Fish84a} and has recently attracted significant attention \citep{Bran13a,BBS20a,Peyr13a,Hoan17a}.\footnote{In France, maximal lotteries have been popularized under the name \emph{scrutin de Condorcet randomis\'e} (randomized Condorcet voting system).}

In this paper, we explore the limitations of collective choice with \pc preferences, in particular with regard to strategyproofness, participation, and efficiency. More specificially, we are interested in the question of whether there are attractive SDSs that satisfy \pc-strategyproofness or strict \pc-participation. The latter axiom demands that a voter is \emph{strictly} better off participating unless he is already at maximum happiness.
We prove the following theorems, all of which settle open problems raised by \citet[p.~18]{Bran17a}:

\begin{itemize}
    \item There is no Condorcet-consistent SDS that satisfies \pc-strategyproofness (\Cref{thm:CCSP}).
    \item There is no anonymous and neutral SDS that satisfies \pc-efficiency and \pc-strategyproofness (\Cref{thm:EffSP}).
    \item There is no anonymous and neutral SDS that satisfies \pc-efficiency and strict \pc-participation (\Cref{thm:strictpart}).
\end{itemize}
All three theorems hold for strict preferences and require $m\geq 4$ alternatives; we show that they turn into possibilities when $m\leq 3$ by constructing two new SDSs.
The second theorem strengthens Theorem 5 by \citet{ABBB15a}, which shows an analogous statement for weak preferences.\footnote{However, when ties are allowed in the voters' preferences, much stronger results hold: \citet{BBEG16a} have shown an analogous claim based on \sd preferences, which implies the result by \citet{ABBB15a}. Recent results \citep[e.g.,][]{BBPS21a,BBL21b} hint at the fact that even stronger impossibilities may hold for weak preferences.}

In the appendix, we furthermore settle an open problem concerning \pc-efficiency raised by \citet{ABB14b}: we construct a preference profile and a \pc-inefficient lottery $p$ such that no sequence of \pc-improvements starting from $p$ leads to a \pc-efficient lottery (\Cref{prop:cycle}).

\section{The Model}

Let $A=\{a_1, \dots, a_m\}$ be a finite set of $m$ alternatives and  $\mathbb{N}=\{1,2,3,\dots\}$ an infinite set of voters. We denote by $\mathcal{F}(\mathbb{N})$ the set of all finite and non-empty subsets of $\mathbb{N}$. Intuitively, $\mathbb{N}$ is the set of all potential voters, whereas $N\in\mathcal{F}(\mathbb{N})$ is a concrete electorate. Given an electorate $N\in \mathcal{F}(\mathbb{N})$, every voter $i\in N$ has a \emph{preference relation} $\succ_i$, which is a complete, transitive, and anti-symmetric binary relation on $A$. In particular, we do not allow for ties, which only makes our results stronger. We write preference relations as comma-separated lists and denote the set of all preference relations by $\mathcal{R}$. A \emph{preference profile} $R$ on an electorate $N\in\mathcal{F}(\mathbb{N})$ assigns a preference relation $\succ_i$ to every voter $i\in N$, i.e., $R\in \mathcal{R}^N$. When writing preference profiles, we subsume voters who share the same preference relation. To this end, we define ${\set{j}{k}}=\{i\in N\colon j\leq i\leq k\}$ and note that ${\set{j}{k}}=\emptyset$ if $j>k$. For instance, $\set{1}{3}$: $a,b,c$ means that voters $1$, $2$, and $3$ prefer $a$ to $b$ to $c$. We omit the brackets for singleton sets. Given a preference profile $R\in\mathcal{R}^N$, the \emph{majority margin} between two alternatives $x,y\in A$ is $g_{R}(x,y)=|\{i\in N\colon x\succ_i y\}|-|\{i\in N\colon y\succ_i x\}|$, i.e., the majority margin indicates how many more voters prefer $x$ to $y$ than \emph{vice versa}. Furthermore, we define $n_R(x)$ as the number of voters who prefer alternative~$x$ the most in the profile $R$.
Next, we denote by $R_{-i}=(\succ_1,\dots,\succ_{i-1}, \succ_{i+1},\dots, \succ_n)$ the profile derived from $R\in \mathcal{R}^N$ by removing voter $i\in N$. Finally, we define $\mathcal{R}^{\mathcal{F}(\mathbb{N)}}=\bigcup_{N\in \mathcal{F}(\mathbb{N})} \mathcal{R}^N$ as the set of all possible preference profiles.

The focus of this paper lies on \emph{social decision schemes (SDSs)}, which are functions that map a preference profile to a lottery over the alternatives. A \emph{lottery} $p$ is a probability distribution over the alternatives, i.e., a function $p:A\rightarrow [0,1]$ such that $p(x)\geq 0$ for all $x\in A$ and $\sum_{x\in A} p(x)=1$. The set of all lotteries on $A$ is denoted by $\Delta(A)$. Then, an SDS $f$ formally is a function of the type $f:\mathcal{R}^{\mathcal{F}(\mathbb{N})}\rightarrow \Delta(A)$. We define $f(R,x)$ as the probability assigned to $x$ by $f(R)$ and extend this notion to sets $X\subseteq A$ by letting $f(R,X)=\sum_{x\in X} f(R,x)$. 

In the next sections, we introduce various desirable properties of SDSs. An overview of these axioms and their relationships is given in \Cref{fig:overview}.

\subsection{Fairness and Decisiveness}

Two basic fairness notions are anonymity and neutrality, which require that voters and alternatives are treated equally, respectively. Formally, an SDS $f$ is \emph{anonymous} if $f(\pi(R))=f(R)$ for all electorates $N\in\mathcal{F}(\mathbb{N})$, preference profiles $R\in\mathcal{R}^N$, and permutations $\pi:N\rightarrow N$, where $R'=\pi(R)$ is defined by ${\succ'_i}={\succ_{\pi(i)}}$ for all $i\in N$. 

Analogously, \emph{neutrality} requires that $f(\pi(R))=\pi(f(R))$ for all electorates $N\in\mathcal{F}(\mathbb{N})$, preference profiles $R\in\mathcal{R}^N$, and permutations $\pi:A\rightarrow A$, i.e., $f(\pi(R))$ is equal to the distribution that, for each alternative $x\in A$, assigns probability $f(R,x)$ to alternative~$\pi(x)$. Here, $R'=\pi(R)$ is the profile such that for all $i\in N$ and $x,y\in A$, $\pi(x) \succ_i' \pi(y)$ if and only if $x \succ_i y$. 

A technical condition that many SDSs satisfy is cancellation. An SDS $f$ satisfies \emph{cancellation} if $f(R)=f(R')$ for all preference profiles $R, R'\in \mathcal{R}^{\mathcal{F}(\mathbb{N})}$ such that $R'$ is derived from $R$ by adding two voters with inverse preferences.

A natural further desideratum in randomized social choice concerns the decisiveness of SDSs: randomization should only be necessary if there is no sensible deterministic winner. This idea is, for example, captured in the notion of \emph{unanimity}, which requires that $f(R,x)=1$ for all profiles $R\in\mathcal{R}^{\mathcal{F}(\mathbb{N})}$ and alternatives $x\in A$ such that all voters in $R$ prefer $x$ the most. 

Clearly, this condition is rather weak and there are natural strengthenings, demanding that so-called absolute winners or Condorcet winners need to be returned with probability 1. An \emph{absolute winner} is an alternative $x$ that is top-ranked by more than half of the voters in $R\in \mathcal{R}^N$, i.e., $n_R(x)>\frac{|N|}{2}$. The \emph{absolute winner property} requires that $f(R,x)=1$ for all profiles $R\in\mathcal{R}^{\mathcal{F}(\mathbb{N})}$ with absolute winner $x$.

An alternative $x$ is a \emph{Condorcet winner} in a profile~$R$ if $g_{R}(x,y)>0$ for all $y\in A\setminus \{x\}$. \emph{Condorcet-consistency} requires that the Condorcet winner is chosen with probability $1$ whenever it exists, i.e., $f(R,x)=1$ for all preference profiles $R\in\mathcal{R}^{\mathcal{F}(\mathbb{N})}$ with Condorcet winner~$x$. Since absolute winners are Condorcet winners, Condorcet-consistency implies the absolute winner property, which in turn implies unanimity. 

\subsection{Preferences over Lotteries}

We assume that the voters' preferences over alternatives are lifted to preferences over lotteries via the pairwise comparison (\pc) extension \citep[see, e.g.,][]{ABB14b,ABBB15a,Bran17a,BrBr17a}. 
According to this notion, a voter prefers lottery $p$ to lottery $q$ if the probability that $p$ returns a better outcome than $q$ is at least as large as the probability that $q$ returns a better outcome than $p$, i.e.,
\[p\succsim_i^\pc q \iff \sum_{x,y\in A\colon x\succ_iy} p(x)q(y)\geq \sum_{x,y\in A\colon x\succ_iy} q(x)p(y).\]
The relation $\succsim_i^\pc$ is complete but intransitive (a phenomenon known as the Steinhaus-Trybula paradox). An appealing interpretation of \pc preferences is \emph{ex ante} regret minimization, i.e., given two lotteries, a voter prefers the one which is less likely to result in \emph{ex post} regret.

Despite the simple and intuitive definition, \pc preferences are difficult to work with and cognitively demanding on behalf of the voters because probabilities are multiplied with each other. 
We therefore introduce a variant of the \pc extension where one of the two lotteries under consideration has to be degenerate (i.e., it puts probability $1$ on a single alternative) and any pair of non-degenerate lotteries are deemed incomparable. 
To this end, we define \pcone preferences as follows: a voter $i$ \pcone-prefers lottery $p$ to lottery $q$ if $p\succsim_i^{\pc} q$ and at least one of $p$ and $q$ is degenerate. Assuming that $p(x)=1$ for some $x\in A$, this is equivalent to
\[p\succsim_i^{\pcone} q \iff \sum_{y\in A\colon x\succ_i y} q(y)\geq \sum_{y\in A\colon y\succ_i x} q(y).
\]
In other words, it only needs to be checked whether $q$ puts at least as much probability on alternatives that are worse than $x$ than on those that are better than $x$. This reduces the cognitive burden on voters and suffices for many lottery comparisons in our proofs involving \pc preferences.
It follows immediately from the definition that $p\succsim_i^\pcone q$ implies that $p\succsim_i^\pc q$ for all lotteries $p,q\in\Delta(A)$ and all preference relations $\succ_i$. Note that in contrast to \pc preferences, \pcone preferences are acyclic, i.e., there is no cycle in the strict part of $\succsim_i^{\pcone}$.

The most common way to compare lotteries when only ordinal preferences over alternatives are known is \emph{stochastic dominance (\sd)} \citep[e.g.,][]{Gibb77a,BoMo02a,BBEG16a}:
\[p\succsim_i^\sd q\iff \forall x\in A\colon\sum_{y\in A\colon y\succ_i x} p(y)\geq \sum_{y\in A\colon y\succ_i x} q(y).\]

In contrast to \pc preferences, the relation $\succsim_i^\sd$ is incomplete but transitive. Furthermore, it follows from a result by \citet{Fish84d} that $p\succsim_i^\sd q$ implies $p\succsim_i^\pc q$ for all preference relations $\succ_i$ and all lotteries $p$ and $q$ \citep[see also][]{ABB14b}. In other words, the \sd relation is a subrelation of the \pc relation. We will sometimes leverage this in our proofs because \sd preferences are easier to handle than \pc preferences.

For each $\X\in\{\pc,\pcone,\sd\}$, we say a voter \emph{strictly \X-prefers} $p$ to $q$, denoted by $p\succ_i^\X q$, if $p\succsim_i^\X q$ and not $q\succsim_i^\X p$. Note that both $p\succ_i^\sd q$ and $p\succ_i^\pcone q$ imply $p\succ_i^\pc q$.

For a better understanding of these concepts, consider a voter with the preference relation ${\succ_1}=a,b,c,d$ and three lotteries $p$, $q$, and $r$ with
\begin{align*}
    p(b)=1\text{,}\qquad\qquad
    q(a)=q(d)=\nicefrac{1}{2}\text{, and}\qquad\qquad
    r(b)=r(c)=\nicefrac{1}{2}.
\end{align*}

First, observe that $p\succ_1^\sd r$ since $p$ is derived from $r$ by moving probability from $c$ to~$b$. In particular, this implies that $p\succ_1^\pc r$ and, since $p$ is degenerate, also $p\succ_1^\pcone r$. On the other hand, \sd does not allow to compare $q$ with $p$ or $r$. 
Since $q(a)= q(d)$ and $a\succ_1 b\succ_1 d$, it follows that $q\succsim_1^\pcone p$ and $p\succsim_1^\pcone q$. This means also that $q\succsim_1^\pc p$ and $p\succsim_1^\pc q$. 
The lotteries $q$ and $r$ can only be compared via \pc preferences by computing that 
\[
\sum_{x,y\in A\colon x\succ_1 y} q(x)r(y)=q(a)r(b)+q(a)r(c)=\frac{1}{2}=r(b)q(d)+r(c)q(d)=\sum_{x,y\in A\colon x\succ_1 y} r(x)q(y).
\]
Hence, we have that $q\succsim_1^\pc r$ and $r\succsim_1^\pc q$. Note that this example also shows that $\pc$ preferences are intransitive: $r\succsim_1^\pc q$ and $q\succsim_1^\pc p$ but $p\succ_1^\pc r$.

\subsection{Efficiency}
\label{sec:efficiency}

Next, we discuss efficiency, which intuitively requires that we cannot make a voter better off without making another voter worse off. Since this axiom requires voters to compare lotteries, we need an underlying lottery extension and we therefore define efficiency depending on $\X\in \{\pc, \pcone, \sd\}$. To formalize the intuition behind this property, we say a lottery $p$ \emph{\X-dominates} another lottery~$q$ in a profile $R\in\mathcal{R}^N$ if $p\succsim_i^\X q$ for all voters $i\in N$ and $p\succ_{i^*}^\X q$ for some voter $i^*\in N$. In this case, we also say that $p$ is an \emph{\X-improvement} of $q$. Less formally, $p$ is an \X-improvement of $q$ if $p$ makes every voter weakly better off and at least one strictly better. According to our intuition $q$ can therefore not be efficient and we thus say a lottery $p$ is \emph{\X-efficient} in $R$ if it is not \X-dominated by any other lottery. Similarly, an SDS $f$ is \emph{\X-efficient} if $f(R)$ is \X-efficient for all preference profiles $R\in\mathcal{R}^{\mathcal{F}(\mathbb{N})}$.

Since $p\succsim_i^\sd q$ and $p\succsim_i^\pcone q$ imply $p\succsim_i^\pc q$ for all voters $i$ and lotteries $p$ and $q$, it follows that a lottery that is \sd-dominated or \pcone-dominated is also \pc-dominated. Hence, for every profile $R$, the set of \pc-efficient lotteries is contained in both the sets of \pcone-efficient and \sd-efficient lotteries. This means that \pc-efficiency implies \sd-efficiency and \pcone-efficiency. Moreover, both \pc-efficiency and \sd-efficiency imply the well-known concept of \emph{ex post} efficiency. In order to introduce this concept, we say an alternative $x$ \emph{Pareto-dominates} another alternative $y$ in a profile $R\in\mathcal{R}^N$ if $x\succ_i y$ for all voters $i\in N$. Recall here that ties in $\succ_i$ are not allowed. \emph{Ex post efficiency} then requires that $f(R,x)=0$ for all profiles $R\in\mathcal{R}^{\mathcal{F}(\mathbb{N})}$ and alternatives $x\in A$ that are Pareto-dominated in $R$. 

To illuminate the natural relationship between \emph{ex post} efficiency and \pc-efficiency, let us take a probabilistic view on \emph{ex post} efficiency. First, observe that there is no alternative $x$ that is preferred by all voters to an alternative drawn from an \emph{ex post} efficient lottery $p$. Hence, for any other lottery $q$, the probability that $q$ returns an outcome that is unanimously preferred to an outcome returned by $p$ is $0$, i.e., $\mathbb{P}(\forall i\in N\colon q\succ_i p)=0$, where we view the lotteries $p$ and $q$ as random variables on $A$. Conversely, if $p$ is not \emph{ex post} efficient, it follows that $\mathbb{P}(\forall i\in N\colon q\succ_i p)>0$ for the lottery $q$ derived from $p$ by shifting the probability from the Pareto-dominated alternatives to their dominator. Hence, a lottery $p$ is \emph{ex post} efficient in a profile $R\in\mathcal{R}^N$ if and only if there is no other lottery $q$ such that 
\begin{equation}
    \mathbb{P}(\forall i\in N\colon q\succ_i p)>\mathbb{P}(\forall i\in N\colon p\succ_i q).\label{eq:po}
\end{equation}

From this inequality, one immediately obtains \pc-efficiency by moving the quantification over the voters outside of the probability: a lottery $p$ is is \pc-efficient in a profile $R\in\mathcal{R}^N$ if and only if there is no other lottery $q$ such that
\begin{equation}
    \begin{split}
    & \forall i\in N\colon \mathbb{P}(q\succ_i p)\geq \mathbb{P}(p\succ_i q) \\ \land \quad &  \exists i\in N\colon \mathbb{P}(q\succ_i p)> \mathbb{P}(p\succ_i q).\label{eq:pc}
    \end{split}
\end{equation}

Despite its simple and intuitive definition, \pc-efficiency is surprisingly complex and little understood. For instance, while \citet{ABB14b} prove that the set of \pc-efficient lotteries is non-empty and connected, they also provide examples showing that this set may fail to be convex and can even be ``curved'' (i.e., it is not the union of a finite number of polytopes). Furthermore, they construct a preference profile with a \pc-dominated lottery $p$ that is not dominated by any \pc-efficient lottery. In their example, however, one can find an intermediate lottery which \pc-dominates $p$ and which is \pc-dominated by a \pc-efficient lottery. \citeauthor{ABB14b} conclude their paper by writing that ``it is an interesting open problem whether there always is a path of Pareto improvements from every [\pc-]dominated lottery to some [\pc-]undominated lottery'' \citep[][p.~129]{ABB14b}. In \Cref{subsec:cycle}, we answer this problem in the negative by providing a profile with five alternatives and eight voters, where following \emph{any} sequence of \pc-improvements from a certain lottery $p$ will lead back to $p$ and it is thus not possible to reach a \pc-efficient outcome by only applying \pc-improvements.

\subsection{Incentive-Compatibility}

The final class of axioms we consider consists of strategyproofness, and (strict) participation. Just like efficiency, both of these axioms can be defined for all lottery extensions; we thus define the concepts for $\X\in\{\pc,\pcone,\sd\}$. 

\paragraph{Strategyproofness.} Intuitively, strategyproofness demands that no voter can benefit by lying about his true preferences. Since \sd and \pcone are incomplete, there are two different ways of defining this axiom depending on how incomparable lotteries are handled. The first option, which we call \emph{\X-strategyproofness}, requires of an SDS $f$ that $f(R)\succsim^\X_i f(R')$ for all electorates $N\in\mathcal{F}(\mathbb{N})$, voters $i\in N$, and preference profiles $R,R'\in\mathcal{R}^N$ with $R_{-i}=R'_{-i}$. In particular, this means that we interpret a deviation from a lottery $p$ to another lottery $q$ as a manipulation if $p$ is incomparable to $q$ with respect to \X. 

\X-strategyproofness is predominant in the literature on \sd preferences \citep[e.g.,][]{Gibb77a,Barb79b,Barb10a}, but it becomes very prohibitive for sparse preference relations over lotteries. For instance, not even the SDS which always returns the uniform lottery over the alternatives is \pcone-strategyproof because \pcone cannot compare the uniform lottery to itself. For such preferences, the notion of weak \X-strategyproofness is more sensible: an SDS $f$ is \emph{weakly \X-strategyproof} if $f(R')\not\succ_i^{\X} f(R)$ for all electorates $N\in\mathcal{F}(\mathbb{N})$, voters $i\in N$, and preference profiles $R,R'\in\mathcal{R}^N$ with $R_{-i}=R'_{-i}$. Or, in words, weak strategyproofness requires that no voter can obtain a strictly \X-preferred outcome by lying about his true preferences.\footnote{In the literature, \X-strategyproofness is sometimes called strong \X-strategyproofness, and weak \X-strategyproofness is then called \X-strategyproofness. This is, for instance, the case in the survey by \citet{Bran17a}.} Note that, since \pc preferences are complete, \pc-strategyproofness coincides with weak \pc-strategyproofness. By contrast, for \sd and \pcone, weak \X-strategyproofness is significantly less demanding than \X-strategyproofness. 
We say an SDS is \emph{\X-manipulable} if it is not \X-strategyproof and \emph{strongly \X-manipulable} if it is not weakly \X-strategyproof. Furthermore, since strategyproofness does not require a variable electorate, we always specify the electorates for which we show that an SDS is strategyproof or manipulable. 

\paragraph{Participation.} As the last point of this section, we discuss participation axioms, which intuitively require that voters should not be able to benefit by abstaining from the election. Analogous to strategyproofness, one could formalize this condition in two ways depending on how incomparabilities between lotteries are interpreted.\footnote{As with strategyproofness, both versions are equivalent for \pc because \pc preferences are complete.} Nevertheless, we will focus in our results only on the strong notion and thus say that an SDS $f$ satisfies \emph{\X-participation} if $f(R)\succsim_i^\X f(R_{-i})$ for all electorates $N\in\mathcal{F}(\mathbb{N})$, voters $i\in N$, and preference profiles $R\in\mathcal{R}^N$. 

In this paper, we are mainly interested in \emph{strict \X-participation}, as introduced by \citet{BBH15b}, which demands of an SDS $f$ that, for all $N\in\mathcal{F}(\mathbb{N})$, $i\in N$, and $R\in\mathcal{R}^N$, it holds that $f(R)\succsim_i^\X f(R_{-i})$ and, moreover, $f(R)\succ_i^\X f(R_{-i})$ if there is a lottery $p$ with $p\succ_i^\X f(R_{-i})$. That is, whenever possible, a voter \emph{strictly} benefits from voting compared to abstaining.\bigskip

Since both $p\succsim_i^\sd q$ and $p\succsim_i^\pcone q$ imply $p\succsim_i^\pc q$, the concepts above are related for \pc, \pcone, and \sd: \sd-strategyproofness implies \pc-strategyproofness which implies weak \pcone-strategyproofness. Furthermore, strict \sd-participation is stronger than strict \pc-participation, which obviously entails \pc-participation \citep[cf.][]{Bran17a}. An overview of these relationships is given in \Cref{fig:overview}. 

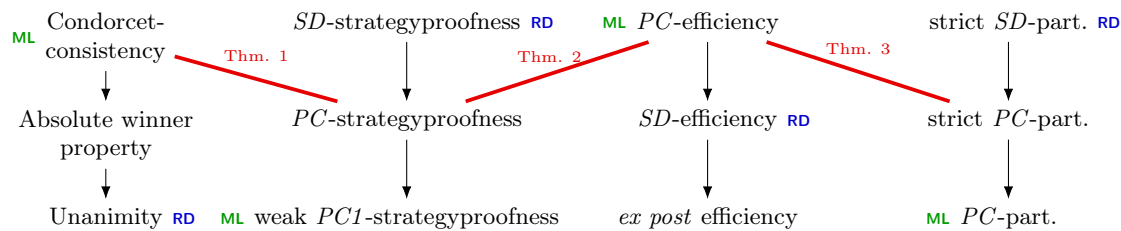
\begin{figure*}[tb]
    \centering
\begin{tikzpicture}\footnotesize
\tikzstyle{arrow}=[->,>=angle 60, shorten >=1pt,draw]
\tikzstyle{mynode}=[align=center, anchor=north]
\tikzstyle{myedge}=[color=red!90!black, line width=0.05cm]
\tikzstyle{ml}=[label={[label distance=-0.1cm]180:\textsf{\tiny\bfseries\color{green!60!black}ML}}]
\tikzstyle{rd}=[label={[label distance=-0.1cm]0:\textsf{\tiny\bfseries\color{blue!80!black}RD}}]

\def\dis{1.3} 
		\def\decisivecol{0} 
		\def\partcol{12} 
		\def\spcol{4} 
		\def\efficiencycol{8} 
	
		\draw (\efficiencycol,0) node[mynode, ml] (PC-eff) {\pc-efficiency}
		++(0,-1*\dis) node[mynode,rd](SD-eff) {\sd-efficiency}
		++(0,-1*\dis) node[mynode] (po) {\emph{ex post} efficiency};
		
		\draw[-Latex] (PC-eff) to (SD-eff);
		\draw[-Latex] (SD-eff) to (po);

		\draw (\spcol,0) node[mynode,rd] (SD-sp) {\sd-strategyproofness}
		++(0,-1*\dis) node[mynode] (PC-sp) {\pc-strategyproofness}
		++(0,-1*\dis) node[mynode,ml] (PC1-sp) {weak \pcone-strategyproofness};
		
		\draw[-Latex] (SD-sp) to (PC-sp);
		\draw[-Latex] (PC-sp) to (PC1-sp);
		
		\draw (\partcol,0) node[mynode,rd] (s-SD-part) {strict \sd-part.}
		++(0,-1*\dis) node[mynode] (s-PC-part) {strict \pc-part.}
		++(0,-\dis) node[mynode, ml] (PC-part) {\pc-part.};
		
		\draw[-Latex] (s-SD-part) to (s-PC-part);
		\draw[-Latex] (s-PC-part) to (PC-part);
		
		\draw (\decisivecol,0) node[mynode, ml] (CC) {Condorcet-\\consistency}
		++(0,-1*\dis) node[mynode] (AWP) {Absolute winner\\property}
		++(0,-1*\dis) node[mynode,rd] (Una) {Unanimity};

        \draw[-Latex] (CC) to (AWP);
        \draw[-Latex] (AWP) to (Una);
        
        \draw[myedge] (CC) edge node[above, yshift=3pt ] {\tiny Thm. \ref{thm:CCSP}} (PC-sp);
        
        \draw[myedge] (PC-sp) edge node[above, xshift=-3pt] {\tiny Thm. \ref{thm:EffSP}} (PC-eff);
        
        \draw[myedge] (PC-eff) edge node[above, yshift=3pt ] {\tiny Thm. \ref{thm:strictpart}} (s-PC-part);
		
\end{tikzpicture}
    \caption{Overview of results. An arrow from an axiom $X$ to another axiom $Y$ indicates that $X$ implies $Y$. The red lines between axioms represent impossibility theorems. Note that \Cref{thm:EffSP,thm:strictpart} additionally require anonymity and neutrality. Axioms labeled with \ml are satisfied by maximal lotteries, and axioms labeled with \rd are satisfied by the uniform random dictatorship.
    }
    \label{fig:overview}
\end{figure*}

\section{Random Dictatorship and Maximal Lotteries}
\label{sec:RD-ML}

The following two important SDSs are useful for putting our results into perspective: the uniform random dictatorship (\rd) and maximal lotteries (\ml). These SDSs are well-known and most of the subsequent claims are taken from the survey by \citet{Bran17a}. The \emph{uniform random dictatorship (\rd)} assigns probabilities proportional to $n_R(x)$, i.e., $\rd(R,x)=\frac{n_R(x)}{\sum_{y\in A}n_R(y)}$ for every alternative $x\in A$ and preference profile $R\in\mathcal{R}^{\mathcal{F}(\mathbb{N})}$. More intuitively, \rd chooses a voter uniformly at random and returns his favorite alternative as the winner. This SDS satisfies strong incentive axioms, but fails efficiency and decisiveness criteria. 

\begin{proposition}\label{prop:rd}
\rd satisfies \sd-strategyproofness, strict \sd-participation, and \sd-efficiency, but fails \pcone-efficiency and the absolute winner property.
\end{proposition}

Clearly, since \sd-strategyproofness and strict \sd-participation imply the corresponding concepts for \pc, \rd satisfies these incentive axioms also for \pc preferences. Even more, when additionally requiring anonymity, \rd is the only SDS that satisfies \sd-strategyproofness and \sd-efficiency \citep{Gibb77a}. On the other hand, a result by \citet{Beno02a} implies that no \sd-strategyproof SDS can satisfy the absolute winner property. For \rd, this claim as well as its failure of \pcone-efficiency can be observed in the following profile.

\begin{profile}
$R$: & $\set{1}{3}$: $a,b,c$ & $4$: $b,a,c$ & $5$: $c,a,b$
\end{profile}
For this profile, $\rd(R,a)=\frac{3}{5}$ and $\rd(R,b)=\rd(R,c)=\frac{1}{5}$, but $a$ is the absolute winner and the lottery that puts probability $1$ on $a$ \pcone-dominates $\rd(R)$.

The set of maximal lotteries for profile $R$ is defined as 
\[
\ml(R)= \{p\in \Delta(A)\colon \sum_{x,y\in A} p(x)q(y)g_{R}(x,y)\geq 0 \text{ for all } q\in \Delta(A)\}.
\]
$\ml(R)$ is non-empty by the minimax theorem and almost always a singleton. In particular, if the number of voters is odd, there is always a unique maximal lottery. In case of multiple maximal lotteries, the claim that \ml satisfies weak \pcone-strategyproofness requires a mild tie-breaking assumption: a degenerate lottery may only be returned if it is the unique maximal lottery. For all other claims, ties can be broken arbitrarily.
As the next proposition shows, \ml satisfies strong efficiency and decisiveness notions but is rather manipulable. 

\begin{proposition}\label{prop:ml}
\ml satisfies \pc-efficiency, \pc-participation, Condorcet-consistency, and weak \pcone-strategyproofness, but fails \pc-strategyproofness and strict \pc-participation.
\end{proposition}

References for all claims except the one concerning strict \pc-participation are given by \citet{BBS20a}. The failure of strict \pc-participation is straightforward because \ml is Condorcet-consistent and a voter may be unable to change the Condorcet winner by joining the electorate. \citet{BBS20a} show that \ml is \pc-manipulable in most profiles that admit no weak Condorcet winner. This is, for example, the case in the profiles $R$ and $R'$ below, where Voter $4$ can \pc-manipulate by deviating from $R$ to $R'$. 

\begin{profile}
$R$: & $\{1,2\}$: $a,b,c$ & $\{3,4\}$: $b,c,a$ & $5$: $c,a,b$\\
$R'$: & $\{1,2\}$: $a,b,c$ & $3$: $b,c,a$ & $\{4,5\}$: $c,a,b$
\end{profile}

The unique maximal lotteries in $R$ and $R'$, respectively, are $p$ and $q$ with $p(a)=q(c)=\frac{3}{5}$ and $p(b)=p(c)=q(a)=q(b)=\frac{1}{5}$. Since ${\succ_i}={\succ_i'}$ for all $i\in \{1,2,3,5\}$ and $q \succ_4^\pc p$, Voter $4$ can \pc-manipulate by deviating from $R$ to $R'$.

On the other hand, \ml is very efficient. In fact, maximality of lotteries can be seen as an efficiency notion itself. To this end, note that a lottery $p$ is in $\ml(R)$ if and only if there is no lottery $q$ such that a voter that is uniformly drawn from $N$ is more likely to prefer an outcome drawn from $q$ to an outcome drawn from $p$ than vice versa \citep{BrBr17a}. More formally, let $I$ denote a uniformly distributed random variable on the voters, and interpret $p$ and $q$ as independent random variables on the alternatives. Then, $p\in \ml(R)$ if and only if there is no lottery $q$ such that 
\begin{equation}
    \mathbb{P}(q\succ_I p)>\mathbb{P}(p\succ_I q).
\end{equation}

This condition is equivalent to the definition of $\ml$ because 
\[
\mathbb{P}(q\succ_I p)=\sum_{x,y\in A} q(x)p(y)\frac{|\{i\in N\colon x\succ_i y\}|}{|N|}
\]
and thus $\mathbb{P}(q\succ_I p)>\mathbb{P}(p\succ_I q)$ if and only if $\sum_{x,y\in A} q(x)p(y)g_R(x,y)>0$. When comparing this to the definitions of \pc-efficiency and \emph{ex post} efficiency presented in \Cref{eq:po,eq:pc}, it can be seen that the ``only'' difference between maximal lotteries and \pc-efficient lotteries is that for maximal lotteries a voter is uniformly drawn at random while for \pc-efficiency the inequality has to hold for all voters. This immediately implies that \ml satisfies \pc-efficiency.

\section{Results}

We are now ready to present our results. 
The results for \pc-strategyproofness are given in \Cref{subSec:SP} while those for strict \pc-participation are given in \Cref{subSec:Part}. For the sake of readability, we defer all lengthy proofs to the appendix.

\subsection[PC-strategyproofness]{\pc-strategyproofness}\label{subSec:SP}

In this section, we show that every Condorcet-consistent and every anonymous, neutral, and \pc-efficient SDS is \pc-manipulable when there are $m\geq 4$ alternatives. These results show that no SDS simultaneously satisfies \pc-strategyproofness and some of the desirable properties of maximal lotteries. Moreover, since \pc-strategyproofness is weaker than \sd-strategyproofness, the incompatibility of \pc-strategyproofness and Condorcet-consistency is a strengthening of the well-known impossibility of Condorcet-consistent and \sd-strategyproof SDSs.
The second results is somewhat surprising: while anonymity, neutrality, \sd-strategyproofness, and \sd-efficiency characterize the uniform random dictatorship, 
the axioms become incompatible when moving from \sd to \pc. Both impossibilities require $m\geq 4$ alternatives and we show that they turn into possibilities when $m\leq 3$. 

\begin{restatable}{theorem}{CCSP}\label{thm:CCSP}
Every Condorcet-consistent SDS is \pc-manipulable if $|N|\geq 5$ is odd and $m\geq 4$.
\end{restatable}
\begin{proof}
		Assume for contradiction that there is a Condorcet-consistent and \pc-strategyproof SDS $f$ for $m\geq 4$ alternatives. Subsequently, we focus on the electorate $N=\{1,\dots, 5\}$ because we can generalize the result to any larger electorate with an odd number of voters by adding pairs of voters with inverse preferences. These voters do not change the Condorcet winner and hence will not affect our analysis.
		
		As the first step, consider the profiles $R^1$ to $R^4$. The $*$ symbol is a placeholder for all missing alternatives.

	    \begin{profile}
	          $R^{1}$: & $1$: $a,b,d,c,*$ & $2$: $d, b, a, c, *$ & $3$: $a,d,c,b,*$
	          & $4$: $*,c,d,b,a$ & $5$: $c,b,a,d,*$\\
			  $R^{2}$: & $1$: $a,b,d,c,*$ & $2$: $d, b, a, c, *$ & \color{black}{$3$: $a,b,c,d,*$}
			  & $4$: $*,c,d,b,a$ & $5$: $c,b,a,d,*$\\
			  $R^{3}$: & $1$: $a,b,d,c,*$ & $2$: $d, b, a, c, *$ & $3$: $a,d,c,b,*$
			  & \color{black}{$4$: $*,c,d,a,b$} & $5$: $c,b,a,d,*$\\
			  $R^{4}$: & $1$: $a,b,d,c,*$ & $2$: $d, b, a, c, *$ & \color{black}{$3$: $d,a,c,b,*$}
			  & $4$: $*,c,d,b,a$ & $5$: $c,b,a,d,*$
	    \end{profile}
	    
		Note that $b$ is the Condorcet winner in $R^2$, $a$ in $R^3$, and $d$ in $R^4$. Thus, Condorcet-consistency entails that $f(R^2,b)=f(R^3,a)=f(R^4,d)=1$. In contrast, there is no Condorcet winner in $R^1$ and we use \pc-strategyproofness to derive $f(R^1)$. In more detail, since $f(R^2)$, $f(R^3)$, and $f(R^4)$ are degenerate, it suffices to use weak \pcone-strategyproofness. For instance, this axiom postulates that 
		$\sum_{x\in A\colon x\succ_3^2 b} f(R^1,x) \leq \sum_{x\in A\colon b\succ_3^2 x} f(R^1,x)$, as otherwise Voter $3$ can \pc-manipulate by deviating from $R^2$ to $R^1$. Equivalently, this means that
		\begin{align}\label{eq:1}
			&f(R^1,a)\leq f(R^1,A\setminus \{a,b\}).
		\end{align}
		Analogously, weak \pcone-strategyproofness between $R^3$ and $R^1$ and between $R^1$ and $R^4$ entails the following inequalities because Voter $4$ in $R^3$ needs to \pcone-prefer $f(R^3)$ to $f(R^1)$ and Voter~$3$ in $R^1$ needs to \pcone-prefer $f(R^1)$ to $f(R^4)$.
		\begin{align}
			&f(R^1,A\setminus\{a,b\}) \leq f(R^1,b)\label{eq:2}\\
			&f(R^1,A\setminus\{a,d\}) \leq f(R^1,a)\label{eq:3}
		\end{align}
		
	Chaining the inequalities together, we get
	$f(R^1,A\setminus\{a,d\}) \leq f(R^1,a) \leq f(R^1,A\setminus \{a,b\}) \leq f(R^1,b)$, so $f(R^1, A\setminus\{a,b,d\})=0$. Simplifying (\ref{eq:1}),  (\ref{eq:2}), and (\ref{eq:3}) then results in  $f(R^1,a)\leq f(R^1,d)\leq f(R^1,b)\leq f(R^1,a)$, so $f(R^1,a)=f(R^1,b)=f(R^1,d)=\frac{1}{3}$.
	
	Next, we analyze the profiles $R^5$ to $R^8$.
	
    \begin{profile}
          $R^{5}$: & $1$: $a,b,d,c,*$ & $2$: $b, d, a, c, *$ & $3$: $a,d,c,b,*$ 
          & $4$: $*,c,d,b,a$ & $5$: $c,b,a,d,*$\\

		  $R^{6}$: & $1$: $a,b,d,c,*$ & $2$: $b, d, a, c, *$ & $3$: $a,d,c,b,*$ 
		  & $4$: $*,c,d,b,a$ & \color{black}{$5$: $b,c,a,d,*$}\\
		  $R^{7}$: & $1$: $a,b,d,c,*$ & $2$: $b, d, a, c, *$ & $3$: $a,d,c,b,*$ 
		  & $4$: $*,c,d,a,b$ & $5$: $c,b,a,d,*$\\

		  $R^{8}$: & $1$: $a,b,d,c,*$ & $2$: $b, c, d, a, *$ & $3$: $a,d,c,b,*$ 
		  & $4$: $*,c,d,b,a$ & $5$: $c,b,a,d,*$
    \end{profile}
	Just as for the profiles $R^1$ to $R^4$, there is no Condorcet winner in $R^5$, whereas $b$ is the Condorcet winner in $R^6$, $a$ in $R^7$, and $c$ in $R^8$. Consequently, Condorcet-consistency requires that $f(R^6,b)=f(R^7,a)=f(R^8,c)=1$. Next, we use again weak \pcone-strategyproofness to derive $f(R^5)$. In particular, we infer the following inequalities as Voter $5$ in $R^5$ needs to \pcone-prefer $f(R^5)$ to $f(R^6)$, Voter $4$ in $R^7$ needs to \pcone-prefer $f(R^7)$ to $f(R^5)$, and Voter $2$ in $R^8$ needs to \pcone-prefer $f(R^8)$ to $f(R^5)$. 
	\begin{align}
		&f(R^5,A\setminus \{b,c\})\leq f(R^5,c)\\
		&f(R^5,A\setminus \{a,b\})\leq f(R^5,b)\\
		&f(R^5, b) \leq f(R^5,A\setminus \{b,c\})
	\end{align}
	
    Analogous computations as for $R^1$ now show that $f(R^5,a)=f(R^5,b)=f(R^5,c)=\frac{1}{3}$.
    Finally, note that $R^1$ and $R^5$ only differ in the preference of Voter $2$. This means that Voter $2$ can \pc-manipulate by deviating from $R^5$ to $R^1$ since he even \sd-prefers $f(R^1)$ to $f(R^5)$. Hence, $f$ fails \pc-strategyproofness, which contradicts our assumptions. 
\end{proof}

Before proving the incompatibility of \pc-efficiency and \pc-strategyproofness, we first state an auxiliary claim which establishes that the absolute winner property, \pc-efficiency, and \pc-strategyproofness are incompatible. The involved proof of this lemma is deferred to \Cref{subsec:EffSP}.

\begin{restatable}{lemma}{AWPSP}\label{lem:AWPSP}
Every \pc-efficient SDS that satisfies the absolute winner property is \pc-manipulable if $|N|\geq 3$, $|N|\not\in\{4,6\}$, and $m\geq 4$. 
\end{restatable}

Note that \Cref{lem:AWPSP} is a rather strong impossibility itself and, in particular, does not require anonymity or neutrality. Based on this lemma, we now show that every anonymous, neutral, and \pc-efficient SDS is \pc-manipulable. It is sufficient to show that the given axioms imply the absolute winner property since the result then follows from \Cref{lem:AWPSP}.

\begin{restatable}{theorem}{EffSP}\label{thm:EffSP}
Every anonymous and neutral SDS that satisfies \pc-efficiency is \pc-manipulable if $|N|\geq 3$, $|N|\not\in\{4,6\}$, and $m\geq 4$.
\end{restatable}
\begin{proof}
We prove the claim for even $|N|$; the argument for odd $|N|$ is much more involved and deferred to the appendix. Our goal is to show that every SDS that satisfies \pc-efficiency, \pc-strategyproofness, anonymity, and neutrality for an electorate with an even number of voters $|N|\geq 8$ also satisfies the absolute winner property. Then, \Cref{lem:AWPSP} shows that no such SDS exists. To this end, suppose that there is an SDS $f$ that satisfies all given axioms and consider the following profile $R^1$, where the $*$ symbol indicates that all missing alternatives are added in an arbitrary fixed order.

\begin{profile}
       $R^1$: & $1$: $a,b,c,*$ & $2$: $a,c,b,*$
	    & $\set{3}{\frac{n}{2}+1}$: $b,a,c,*$ & $\set{\frac{n}{2}+2}{n}$: $c,a,b,*$
\end{profile}

First, all alternatives except $a$, $b$, and $c$ are Pareto-dominated and thus \pc-efficiency requires that $f(R^1,x)=0$ for $x\not\in \{a,b,c\}$. Moreover, $b$ and $c$ are symmetric in $R^1$ and anonymity and neutrality therefore imply that $f(R^1,b)=f(R^1,c)$. Finally, note that every lottery $p$ with $p(b)=p(c)>0$ is \pc-dominated by the lottery $q$ with $q(a)=1$. Hence, it follows from \pc-efficiency, anonymity, and neutrality that $f(R^1,a)=1$. 

Next, consider the profile $R^2$, in which the voters in $[3\dots \frac{n}{2}+1]$ report $a,c,b$ instead of $b,a,c$, and Voter $1$ reports $a,c,b$ instead of $a,b,c$. 

\begin{profile}
       $R^2$: & $1$: $a,c,b,*$ & $2$: $a,c,b,*$
	    & $\set{3}{\frac{n}{2}+1}$: $a,c,b,*$ & $\set{\frac{n}{2}+2}{n}$: $c,a,b,*$
\end{profile}

A repeated application of \pc-strategyproofness shows that $a$ must still be chosen with probability $1$ in $R^2$ because $a$ is the favorite alternative of the deviator after every step. In more detail, if $a$ is assigned probability $1$ in profile $R$ and a voter deviates to a profile $R'$ by top-ranking $a$, $a$ must still have probability $1$ because otherwise, the voter can manipulate in $R'$ by going back to $R$. 

Finally, observe that $c$ Pareto-dominates all alternatives but $a$ in $R^2$. Using this fact, we go to the profile $R^3$ by letting the voters in $\set{\frac{n}{2}+2}{n}$ one after another change their preference relation to $c,b,*,a$.

\begin{profile}
       $R^3$: & $1$: $a,c,b,*$ & $2$: $a,c,b,*$
	    & $\set{3}{\frac{n}{2}+1}$: $a,c,b,*$ & $\set{\frac{n}{2}+2}{n}$: $c,b,*,a$
\end{profile}

We claim that $f(R^3,a)=1$. Indeed, \pc-efficiency shows for $R^3$ and all intermediate profiles that only $a$ and $c$ can have positive probability as all other alternatives are Pareto-dominated. Moreover, \pc-strategyproofness shows for every step that, if $a$ is originally chosen with probability $1$, then $c$ must have probability $0$ after the manipulation because every manipulator prefers $c$ to $a$ and no other alternative gets any positive probability. Hence, we infer that $f(R^3,a)=1$.

Finally, note that the voters who top-rank $a$ can now reorder the alternatives in $A\setminus \{a\}$ arbitrarily and the voters who bottom-rank $a$ can even reorder all alternatives without affecting the outcome.  In more detail, if $a$ is chosen with probability $1$ and a voter top-ranks $a$ after the manipulation, \pc-strategyproofness requires that $a$ still is assigned probability $1$ because the voter can otherwise manipulate by switching back to his original preference relation. Similarly, if a voter bottom-ranks alternative $a$ and $a$ is assigned probability $1$, he cannot affect the outcome by deviating because any other outcome induces a \pc-manipulation. Hence, it follows that $f(R,a)=1$ for all profiles in which the voters in $[1\dots\frac{n}{2}+1]$ top-rank $a$. Since anonymity allows us to rename the voters and neutrality to exchange the alternatives, this means that $f$ satisfies the absolute winner property.
\end{proof}

Since both \Cref{thm:CCSP,thm:EffSP} require $m\geq 4$ alternatives, there is still hope for a possibility when $m\leq 3$. Indeed, for $m=2$, 
\ml satisfies Condorcet-consistency, \pc-efficiency, \pc-strategyproofness, anonymity, and neutrality. However, as shown in \Cref{sec:RD-ML}, \ml fails \pc-strategyproofness when $m= 3$.
We therefore construct another SDS that satisfies all given axioms. To this end, let $\textit{CW}(R)$ be the set of Condorcet winners in $R$, and $\textit{WCW}(R)=\{x\in A\colon g_R(x,y)\geq 0 \text{ for all }y\in A\setminus \{x\}\}$ the set of weak Condorcet winners if $\textit{CW}(R)=\emptyset$, and $\textit{WCW}(R)=\emptyset$ otherwise. Then, define the SDS $f^1$ as follows.
\[f^1(R)=\begin{cases}
[x:1] &\text{if } \textit{CW}(R)=\{x\}\\
[x:\frac{1}{2}; y:\frac{1}{2}]&\text{if } \textit{WCW}(R)=\{x,y\}\\
[x:\frac{3}{5}; y:\frac{1}{5}; z:\frac{1}{5}] &\text{if } \textit{WCW}(R)=\{x\}\\
[x:\frac{1}{3}; y:\frac{1}{3}; z:\frac{1}{3}] &\text{otherwise}
\end{cases}
\]

Note that, in the absence of majority ties, $f^1$ boils down to the rather natural SDS that selects the Condorcet winner with probability~1 and returns the uniform lottery otherwise. This SDS was already proposed by \citet{Pott70a} to achieve strategyproofness in the case of three alternatives. As we show, $f^1$ extends this SDS to profiles with majority ties while preserving a number of desirable properties. In particular, $f^1$ is the \emph{only} SDS for $m=3$ alternatives that satisfies cancellation and the axioms of \Cref{thm:EffSP}. We defer the proof of this claim to \Cref{subsec:propositions}. Moreover, $f^1$ is clearly Condorcet-consistent. 

\begin{restatable}{proposition}{possp}\label{prop:m3sp}
For $m=3$, $f^1$ is the only anonymous and neutral SDS that satisfies \pc-efficiency, \pc-strategyproofness, and cancellation. 
\end{restatable}

\begin{remark}
All axioms are required for \Cref{thm:CCSP} and all axioms with the possible exception of neutrality are required for \Cref{thm:EffSP}. \ml only fails \pc-strategyproofness, dictatorships only fail anonymity and Condorcet-consistency, and the uniform random dictatorship only fails \pc-efficiency and Condorcet-consistency. 
The number of alternatives required for \Cref{thm:CCSP,thm:EffSP} is tight as shown by $f^1$.
We conjecture that \Cref{thm:EffSP} holds even without neutrality. 
\end{remark}

\begin{remark}
It is open whether \Cref{thm:CCSP} also holds for even $|N|$. However, when additionally requiring the mild condition of homogeneity (which requires that splitting each voter into $k$ clones with the same preferences does not affect the outcome), the statement holds also for even $|N|\geq 10$. This can be shown by duplicating all voters in the proof and adding some additional profiles in the derivation from $R^1$ to $R^5$.
\end{remark}

\begin{remark}
For the well-known class of pairwise SDSs, which only depend on the majority margins to compute the outcomes, \pc-strategyproofness, unanimity, and homogeneity imply Condorcet-consistency. This follows by carefully inspecting the proof of Lemma 12 by \citet{BrLe21a}. Hence, \Cref{thm:CCSP} implies that no pairwise SDS satisfies unanimity, \pc-strategyproofness, and homogeneity.
\end{remark}

\begin{remark}
Remarkably, the proof of \Cref{thm:CCSP} never uses the full power of \pc-strategyproofness. Instead, every step either uses weak \pcone-strategyproofness or weak \sd-strategyproofness and thus, our proof shows actually a stronger but more technical result where \pc-strategyproofness is replaced with weak \sd-strategyproofness and weak \pcone-strategyproofness.
Interestingly, the Condorcet rule, which chooses the Condorcet winner if it exists and randomizes uniformly over all alternatives otherwise, is Condorcet-consistent and weakly \sd-strategyproof, and \ml is Condorcet-consistent and weakly \pcone-strategyproof. 
\end{remark}

\subsection[Strict PC-participation]{Strict \pc-participation}\label{subSec:Part}

In this section, we show that strict \pc-participation is incompatible with \pc-efficiency. 

\begin{theorem}
\label{thm:strictpart}
No anonymous and neutral SDS satisfies both \pc-efficiency and strict \pc-participation if $m\geq 4$.
\end{theorem}
\begin{proof}
We establish a stronger statement with \pcone-efficiency instead of \pc-efficiency.

Assume for contradiction that there is a neutral and anonymous SDS $f$ that satisfies both \pcone-efficiency and strict \pc-participation. We first prove the impossibility for the case $m=4$ and explain how to generalize the impossibility to more alternatives at the end of this proof. First, consider the following profile with ten voters.

    \begin{profile}
          $R^{1}$: & $1$: $a,b,c,d$ & $2$: $a,b,d,c$ & $3$: $a,c,b,d$
          & $4$: $a,c,d,b$ & $5$: $a,d,b,c$ \\ &$6$: $a,d,c,b$ 
          & $7$: $b,a,c,d$ & $8$: $b,a,d,c$ & $9$: $c,a,b,d$
          & $10$: $c,a,d,b$
    \end{profile}
    
    Observe that $b$ and $c$ are symmetric in $R^1$ and thus, anonymity and neutrality imply that $f(R^1,b)=f(R^1,c)$. Moreover, since $a$ Pareto-dominates $d$, it can be checked that every lottery $p$ with $p(b)=p(c)>0$ or $p(d)>0$ is \pcone-dominated by the lottery $q$ that puts probability $1$ on $a$. Indeed, voters $1$ to $6$ strictly \pcone-prefers $q$ to $p$ since $a$ is their favorite alternative, and voters $7$ to $10$ \pcone-prefer $q$ to $p$ since $p(b)=p(c)$. Hence, \pcone-efficiency requires that $f(R^1,b) = f(R^1,c) = f(R^1,d)=0$, which means that $f(R^1,a) = 1$.

Next, consider profile $R^2$, which is obtained by adding voter~$11$ with the preference relation $d,a,b,c$ to $R^1$. We infer from strict \pc-participation that $f(R^2,d) > f(R^2,b) + f(R^2,c)$.

Finally, consider profile $R^3$, which is obtained by adding voter~$12$ with the preference relation $d,a,c,b$ to $R^2$.
Observe that $b$, $c$, and $d$ are symmetric in $R^3$, so by neutrality and anonymity, $f(R^3,b) = f(R^3,c) = f(R^3,d)$.
If $f(R^3,b) = f(R^3,c) = f(R^3,d) > 0$, then $f$ is not \pcone-efficient because all voters strictly prefer the degenerate lottery that puts probability $1$ on $a$.
Hence, $f(R^3,b) = f(R^3,c) = f(R^3,d) = 0$, which means that $f(R^3,a) = 1$.
Since $f(R^2,d) > f(R^2,b) + f(R^2,c)$, voter~$12$ has a disincentive to participate in $R^3$, thereby contradicting the strict \pc-participation of $f$.

Finally, for extending this impossibility to more than $m=4$ alternatives, we simply add the new alternatives at the bottom of the preference relations of all voters in a fixed order. \pcone-efficiency, anonymity, and neutrality still require for $R^1$ and $R^3$ that $a$ obtains probability~$1$, and it thus is easy to check that the impossibility still holds. 
\end{proof}

By contrast,  multiple SDSs are known to satisfy \sd-efficiency and strict \sd-participation \citep[see][]{BBH15b}. Moreover, \Cref{thm:strictpart} can be seen as a complement to the work of \citet{BBH15c} who showed that \ml satisfies both \pc-participation and \pc-efficiency. In particular, maximal lotteries satisfy a maximal degree of participation subject to \pc-efficiency. 

Since \Cref{thm:strictpart} requires $m\ge 4$, a natural question is whether the impossibility also holds for $m\le 3$. As we demonstrate, this is not the case. If $m=2$, it is easy to see that the uniform random dictatorship satisfies all axioms of \Cref{thm:strictpart}. For $m=3$, however, the uniform random dictatorship fails \pcone-efficiency (see \Cref{sec:RD-ML}).
In light of this, we construct a new SDS that satisfies all axioms used in \Cref{thm:strictpart} and \pc-efficiency. To this end, let $B$ denote the set of alternatives that are never bottom-ranked. Then, the SDS $f^2$ is defined as follows: return the uniform random dictatorship if $|B|\in \{0,2\}$; otherwise (i.e., $|B| = 1$), we delete the alternatives $x\in A\setminus B$ that minimize $n_R(x)$ (if there is a tie, delete both alternatives) and return the outcome of the uniform random dictatorship for the reduced profile. As the following proposition shows, $f^2$ indeed satisfies all axioms of \Cref{thm:strictpart} when $m=3$; the proof is deferred to \Cref{subsec:propositions}.

\begin{restatable}{proposition}{pospart}\label{prop:example-m-3}
For $m = 3$, $f^2$ satisfies anonymity, neutrality, \pc-efficiency, and strict \pc-participation.
\end{restatable}

\begin{remark}
Both \pc-efficiency and \pc-participation are required for \Cref{thm:strictpart} since \ml and \rd satisfy all but one of the axioms. 
Whether anonymity and neutrality are required is open. 
\end{remark}

\begin{remark}
\label{remark:weak-PCeff}
\Cref{thm:strictpart} still holds when replacing \pc-efficiency with \emph{weak \pc-efficiency} and letting $m\ge 5$. (A lottery fails to be weakly \pc-efficient if there is another lottery in which all voters are \emph{strictly} better off.) The incompatibility with strict \pc-participation can be shown by a proof that uses a preference profile with five alternatives and 18 voters that are joined by 6 further voters, but is otherwise similar to that of \Cref{thm:strictpart}.
\end{remark}

\begin{remark}
$f^2$ also satisfies strict \sd-participation, which shows that efficiency and participation are compatible even in their strongest forms when $m=3$. However, these axioms do not uniquely characterize $f^2$.
\end{remark}

\section{Conclusion}

We have studied incentive properties of social decision schemes (SDSs) based on the pairwise comparison (\pc) lottery extension,
and answered open questions raised by \citet{Bran17a} by proving three strong impossibilities. 
In particular, we showed that \pc-strategyproofness and strict \pc-participation are incompatible with \pc-efficiency and Condorcet-consistency when there are at least four alternatives (see also \Cref{fig:overview}). When there are fewer than four alternatives, the axioms are shown to be compatible via the introduction of two new SDSs.
We also settled an open problem by \citet{ABB14b} by showing that there exist profiles and \pc-inefficient lotteries such that it is not possible to reach a \pc-efficient outcome by repeatedly moving from a \pc-dominated lottery to one of its dominators.

We highlight three important aspects and consequences of our results. First, when moving from the standard approach of stochastic dominance (\sd) to \pc, previously compatible axioms become incompatible. In particular, \Cref{thm:EffSP,thm:strictpart} become possibilities when using \sd preferences since all given axioms are satisfied by the uniform random dictatorship. Secondly, unlike Arrow's impossibility and the No-Show paradox, the Gibbard-Satterthwaite impossibility cannot be circumvented via \pc preferences. 
In fact, our results demonstrate limitations of \pc preferences that stand in contrast to the previous positive findings \citep{BBH15c,BrBr17a}. 
In light of the shown tradeoff between incentive-compatibility and efficiency, two SDSs stand out: the uniform random dictatorship because it satisfies \pc-strategyproofness and strict \pc-participation, and maximal lotteries because it satisfies \pcone-strategyproofness, Condorcet-consistency, \pc-efficiency, and \pc-participation.

There are only few opportunities to further strengthen these results. In some cases, it is unclear whether anonymity or neutrality---two fairness properties that are often considered imperative---are required. Two challenging questions concerning \Cref{thm:EffSP} are whether anonymity can be weakened to non-dictatorship and whether \pc-efficiency can be replaced with weak \pc-efficiency (cf.~\Cref{remark:weak-PCeff}). However, if true, any such statement would require quite different proof techniques.

\section*{Acknowledgements}

This work was supported by the Deutsche Forschungsgemeinschaft under grants \mbox{BR 2312/11-2} and \mbox{BR 2312/12-1}, by the Singapore Ministry of Education under grant number MOE-T2EP20221-0001, and by an NUS Start-up Grant. A preliminary version of this paper appeared in the Proceedings of the 31st International Joint Conference on Artificial Intelligence (Vienna, Austria).

\def\bibfont{\small}

\newpage

\appendix
\section{Omitted Proofs}

In this section, we present the proofs omitted from the main body. In particular, we prove in \Cref{subsec:cycle} that it is possible that no path of \pc-improvements originating from a \pc-inefficient lottery leads to a \pc-efficient lottery. Furthermore, we discuss the proofs of \Cref{thm:EffSP} and \Cref{lem:AWPSP} in \Cref{subsec:EffSP}. Finally, \Cref{subsec:propositions} contains the proofs of \Cref{prop:m3sp,prop:example-m-3}. 

\subsection[PC-Efficiency Cycle]{\pc-Efficiency Cycle}\label{subsec:cycle}

\citet{ABB14b} left as an open problem whether for every \pc-inefficient lottery $p$ there is a sequence of \pc-improvements that leads to a \pc-efficient lottery. We disprove this assertion by giving a profile $R$ and a lottery $p$ such that it is not possible to reach a \pc-efficient lottery $q$ from $p$ by repeatedly applying \pc-improvements.

\begin{proposition}\label{prop:cycle}
	There is a profile $R$ and a lottery $p$ such that no sequence of \pc-improvements that starts at $p$ leads to a \pc-efficient lottery. 
\end{proposition}
\begin{proof}
    Before we present the profile and the lotteries that will constitute our counterexample, we discuss two general claims on \pc preferences. To this end, consider two lotteries $p$ and $q$ and two voters $i$ and $j$ and suppose that both voters $i$ and $j$ prefer $q$ to $p$ according to \pc. By using the definition of \pc preferences and partitioning alternatives also with respect to the preference of voter $j$, we derive the following inequality for voter $i$.
    \begin{align*}
	    &\sum_{x,y\in A\colon x \succ_i y \,\land\, x\succ_j y} q(x)p(y) + \sum_{x,y\in A\colon x \succ_i y \,\land\, y\succ_j x} q(x)p(y)\\
	    \geq &\sum_{x,y\in A\colon x \succ_i y \,\land\, x\succ_j y} p(x)q(y) + \sum_{x,y\in A\colon x \succ_i y \,\land\, y\succ_j x} p(x)q(y)	
	\end{align*}
	
	By exchanging the roles of voter $i$ and $j$, we also infer the subsequent inequality.
	\begin{align*}
	    &\sum_{x,y\in A\colon x \succ_j y \,\land\, x\succ_i y} q(x)p(y) + \sum_{x,y\in A\colon x \succ_j y \,\land\, y\succ_i x} q(x)p(y)\\
	    \geq &\sum_{x,y\in A\colon x \succ_j y \,\land\, x\succ_i y} p(x)q(y) + \sum_{x,y\in A\colon x \succ_j y \,\land\, y\succ_i x} p(x)q(y)
	\end{align*}
	
	Now, by summing up these two inequalities and cancelling common terms, we infer our first key insight: two voters $i$ and $j$ simultaneously \pc-prefer $q$ to $p$ only if
	\begin{align}
	    \sum_{x,y\in A\colon x \succ_i y \,\land\, x\succ_j y} q(x)p(y) 
	    \geq \sum_{x,y\in A\colon x \succ_i y \,\land\, x\succ_j y} p(x)q(y).\label{eq:main1}
	\end{align}
	
	On the other hand, if $\sum_{x,y\in A\colon x \succ_i y \,\land\, x\succ_j y} q(x)p(y)=\sum_{x,y\in A\colon x \succ_i y \,\land\, x\succ_j y} p(x)q(y)=0$, our two initial inequalities simplify to 
	 \begin{align}
	    \sum_{x,y\in A\colon x \succ_i y \,\land\, y\succ_j x} q(x)p(y) 
	    = \sum_{x,y\in A\colon x \succ_i y \,\land\, y\succ_j x} p(x)q(y).\label{eq:main2}
	\end{align}
	
	Using \eqref{eq:main1} and \eqref{eq:main2}, we will show that we cannot reach a \pc-efficient lottery from $p^1$ where $p^1(a)=p^1(b)=\frac{1}{2}$ by only making \pc-improvements according to the profile $R$ shown below.
	
	\begin{profile}
		$R$: & $1$: $b,d,c,a,e$ & $2$: $a,e,c,b,d$ & $3$: $d,c,a,b,e$ & $4$: $e,c,a,b,d$\\
			 & $5$: $b,d,e,c,a$ & $6$: $a,e,d,c,b$ & $7$: $b,e,d,c,a$ & $8$: $a,d,e,c,b$
	\end{profile}
	
	To prove this claim, we proceed in three steps, which essentially show that from $p^1$ we can only go towards the lottery $p^2$ with $p^2(c)=1$, from $p^2$ only towards the lottery~$p^3$ with $p^3(d)=p^3(e)=\frac{1}{2}$, and from $p^3$ only towards $p^1$. Hence, we cycle through \pc-inefficient lotteries and never reach an efficient one.\medskip 
	
	\textbf{Step 1:} As the first step, we show that every lottery $p$ with $p(a)=p(b)>0$ and $p(d)=p(e)=0$ is only \pc-dominated by lotteries $q$ with $q(a)=q(b)<p(a)=p(b)$ and $q(d)=q(e)=0$. For this, we observe that voters $1$ and $2$ only agree that $a$ is preferred to $e$ and that is $b$ is preferred to $d$. By \eqref{eq:main1}, $q$ can \pc-dominate $p$ only if
	\[q(b)p(d)+q(a)p(e)\geq p(b)q(d)+p(a)q(e).\] 
	Since $p(d)=p(e)=0$ and $p(a)=p(b)>0$ by assumption, this inequality can only be true if $q(d)=q(e)=0$. Hence, we have that $q(b)p(d)+q(a)p(e)= p(b)q(d)+p(a)q(e)=0$. \Cref{eq:main2} and the fact that $p(d)=p(e)=q(d)=q(e)=0$ therefore show that 
	\[q(b)p(c)+q(b)p(a)+q(c)p(a)= p(b)q(c)+p(b)q(a)+p(c)q(a).\] 
	
	Finally, since $p(a)=p(b)>0$, it is easy to derive that this equation holds only if $q(a)=q(b)$. Using this fact and $q(d)=q(e)=0$, it follows that voters $3$ and $4$ \pc-prefer $q$ to $p$ only if $q(c)\geq p(c)$. Since $p\neq q$ if $q$ \pc-dominates $p$, we thus infer that $q$ dominates $p$ in $R$ only if $q(a)=q(b)<p(a)=p(b)$ and $q(d)=q(e)=p(d)=p(e)=0$. Finally, it is easy to verify that all voters indeed \pc-prefer such a lottery $q$ to $p$. This argument shows that every lottery $p$ with $p=\lambda p^1+(1-\lambda) p^2$ is only \pc-dominated by another lottery of the same form but with smaller $\lambda$. In particular, if we apply this argument at $p^1$, we eventually have to go to $p^2$ when we aim to reach a \pc-efficient lottery by applying \pc-improvements.\bigskip
	
	\textbf{Step 2:} Next, we prove that every lottery $p$ with $p(a)=p(b)=0$, $p(c)>0$, and $p(d)=p(e)$ is only \pc-dominated by lotteries $q$ with $q(a)=q(b)=0$ and $q(d)=q(e)>p(d)=p(e)$. 
	To this end, consider the preferences of voters $3$ and $4$ and note that these voters only agree on the preferences between $c$ and $a$, $c$ and $b$, and $a$ and $b$. Hence, \eqref{eq:main1} requires that 
	\begin{align*}
		q(c)p(a)+q(c)p(b)+q(a)p(b)\geq p(c)q(a)+p(c)q(b)+p(a)q(b).
	\end{align*}
	
	Since $p(a)=p(b)=0$ and $p(c)>0$ by assumption, this inequality can be true only if $q(a)=q(b)=0$. We hence infer that $q(c)\Big(p(a)+p(b)\Big)+q(a)p(b)= p(c)\Big(q(a)+q(b)\Big)+p(a)q(b)=0$. Consequently, \Cref{eq:main2} and the fact that $p(a)=p(b)=q(a)=q(b)=0$ imply that 
	\begin{align*}
		q(d)p(c)+q(d)p(e)+q(c)p(e)= p(d)q(c)+p(d)q(e)+p(c)q(e).
	\end{align*}
	
    Since $p(d)=p(e)$ by assumption, this equation can hold only if $q(d)=q(e)$. This insight and the fact that $q(a)=q(b)=0$ imply that voters $5$ to $8$ only \pc-prefer $q$ to $p$ if $q(d)=q(e)\geq p(d)=p(e)$. Since $p\neq q$ if $q$ \pc-dominates $p$, this inequality must be strict. Finally, it is not difficult to verify that every lottery $q$ that satisfies $q(a)=q(b)=0$ and $q(d)=q(e)>p(d)=p(e)$ indeed \pc-dominates $p$. In particular, this means that every lottery $p$ with $p=\lambda p^2 + (1-\lambda)p^3$ is only \pc-dominated by a lottery $q$ with $p=\lambda' p^2 + (1-\lambda')p^3$ and $\lambda'<\lambda$. Hence, if we are at $p^2$ and try to find a \pc-efficient lottery by applying \pc-improvements, we will inevitably arrive at $p^3$.\bigskip
	
	\textbf{Step 3:} As the last step, we prove that every lottery $p$ with $p(a)=p(b)$, $p(c)=0$, and $p(d)=p(e)>0$ is only \pc-dominated by lotteries $q$ with $q(a)=q(b)>p(a)=p(b)$, $q(c)=0$, and $q(d)=q(e)$. For this, we consider voters $5$ to $8$ and note that voters $5$ and $6$ agree only on the preference between $d$ and $c$ and between $e$ and $c$. Moreover, the same is true for voters~$7$ and $8$ and we thus derive from \eqref{eq:main1} that 
	 \begin{align*}
		 q(d)p(c)+q(e)p(c)\geq p(d)q(c)+p(e)q(c).
	\end{align*}
	
Using the fact that $p(c)=0$ and $p(d)=p(e)>0$, this inequality can be true only if $q(c)=0$. Hence, we have that $q(d)p(c)+q(e)p(c)= p(d)q(c)+p(e)q(c)=0$. Using \Cref{eq:main2} and the fact that $q(c)=p(c)=0$, we derive the following two equations. The first one corresponds to voters $5$ and $6$, whereas the second one corresponds to voters~$7$ and $8$.
	\begin{align}
		& q(b)\Big(p(d)+p(e)+p(a)\Big)+q(d)\Big(p(e)+p(a)\Big)+q(e)p(a)\nonumber\\
		&= p(b)\Big(q(d)+q(e)+q(a)\Big)+p(d)\Big(q(e)+q(a)\Big)+p(e)q(a)\label{eq:aux1}\\
		     &q(b)\Big(p(e)+p(d)+p(a)\Big)+q(e)\Big(p(d)+p(a)\Big)+q(d)p(a)\nonumber\\
		&= p(b)\Big(q(e)+q(d)+q(a)\Big)+p(e)\Big(q(d)+q(a)\Big)+p(d)q(a)\nonumber
	\end{align}
	
	Subtracting these two equations from each other yields that $q(d)p(e)-q(e)p(d)=p(d)q(e)-p(e)q(d)$. Since $p(d)=p(e)>0$, we derive from this equation that $q(d)=q(e)$. Using this observation and the assumptions that $p(a)=p(b)$ and $p(d)=p(e)>0$, we can infer from \Cref{eq:aux1} that $q(a)=q(b)$. In summary, we therefore have that $q(c)=0$, $q(a)=q(b)$, and $q(d)=q(e)$. Next, voters $1$ and $2$ only \pc-prefer $q$ to $p$ if $q(a)=q(b)\geq p(a)=p(b)$, which means that $q(a)=q(b)>p(a)=p(b)$ because $p\neq q$. Finally, it is easily seen that all voters indeed \pc-prefer $q$ to $p$. This means that a lottery $p$ with $p=\lambda p^3+(1-\lambda) p^1$ is only \pc-dominated by another lottery $q$ with the same form but smaller $\lambda$. As a consequence, when applying \pc-improvements from $p^3$ to reach a \pc-efficient outcome, one will inevitably reach $p^1$.
\end{proof}

\subsection[Proof of Theorem 2]{Proof of \Cref{thm:EffSP}}\label{subsec:EffSP}

Next, we prove \Cref{thm:EffSP} and start by presenting three auxiliary lemmas. In more detail, we initiate our analysis by investigating the consequences of \pc-strategyproofness for lotteries with small support. To this end, let the \emph{support} of a lottery be $\supp(p)=\{x\in A\colon p(x)>0\}$. Our first lemma then focuses on the case where $|\supp(f(R))|\leq 2$ and states that, if the support does not change after a manipulation and the manipulator does not reorder the alternatives in the support, the outcome is not allowed to change.

\begin{lemma}\label{lem:2alts}
    Let $R,R'\in \mathcal{R}^N$ for some $N\in\mathcal{F}(\mathbb{N})$, $i\in N$, and $a,b\in A$ such that $R_{-i}=R_{-i}'$ and $a \succ_i b$ iff $a\succ_i'b$. Then, every \pc-strategyproof SDS $f$ satisfies $f(R)=f(R')$ if $\supp(f(R))\subseteq \{a,b\}$ and $\supp(f(R'))\subseteq \{a,b\}$.
\end{lemma}
\begin{proof}
Let $R, R'\in\mathcal{R}^N$, $N\in\mathcal{F}(\mathbb{N})$, $i\in N$, and $a,b\in A$ be two distinct alternatives. Without loss of generality, we assume that $a \succ_ib$ and $a\succ_i'b$; the case where voter $i$ prefers $b$ to $a$ is symmetric. Moreover, let $f$ denote a \pc-strategyproof SDS and assume that $\supp(f(R))\subseteq \{a,b\}$ and $\supp(f(R'))\subseteq \{a,b\}$. Now, assume for contradiction that $f(R)\neq f(R')$. Since both lotteries only put positive probability on $a$ and $b$, this means either that $f(R,a)<f(R',a)$ and $f(R,b)>f(R',b)$, or that $f(R,a)>f(R',a)$ and $f(R,b)<f(R',b)$. First, suppose that $f(R,a)<f(R',a)$. Then, $f(R')\succ_i^\pc f(R)$, and voter $i$ can thus \pc-manipulate by deviating from $R$ to $R'$. On the other hand, if $f(R,a)>f(R',a)$, voter $i$ can \pc-manipulate by deviating from $R'$ to $R$ as he \pc-prefers $f(R)$ to $f(R')$ with respect to $\succ_i'$. Hence, both cases result in a \pc-manipulation, contradicting the \pc-strategyproofness of $f$. This proves that $f(R)=f(R')$. 
\end{proof}

Next, we analyze \pc-strategyproofness when $|\supp(f(R))|\leq 3$. In this case, only a significantly weaker implication holds: if $\supp(f(R))\subseteq \{a,b,c\}$, $\supp(f(R'))\subseteq \{a,b,c\}$, $f(R,a)<f(R,c)$, and $a\succ_i b \succ_i c$ for some voter $i$, then this voter cannot change the fact that $a$ gets less probability than $c$ in the resulting lottery. 

\begin{lemma}\label{lem:3alts}
	Let $R,R'\in \mathcal{R}^N$ for some $N\in\mathcal{F}(\mathbb{N})$, $i\in N$, and $a,b,c\in A$ such that $R_{-i}=R_{-i}'$ and $a\succ_i b \succ_i c$. Then, every \pc-strategyproof SDS $f$ satisfies $f(R',a)<f(R',c)$ if $f(R,a)<f(R,c)$, $\supp(f(R))\subseteq \{a,b,c\}$, and $\supp(f(R'))\subseteq \{a,b,c\}$.
\end{lemma}
\begin{proof}
	Let $R, R'\in\mathcal{R}^N$ for some $N\in\mathcal{F}(\mathbb{N})$, $i\in N$, and $a,b,c\in A$ be three distinct alternatives such that $R_{-i}=R_{-i}'$ and $a\succ_i b\succ_i c$. Furthermore, consider a \pc-strategyproof SDS $f$ and suppose that $\supp(f(R))\subseteq \{a,b,c\}$ and $\supp(f(R'))\subseteq \{a,b,c\}$. For simplicity, we define $p=f(R)$ and $q=f(R')$ and assume for contradiction that $p(a)<p(c)$ and $q(a)\geq q(c)$.
	Next, we use \pc-strategyproofness to relate $p$ and $q$. In particular, we infer the following equation from the \pc-strategyproofness between $R$ and $R'$. Note that the alternatives $x\in A\setminus\{a,b,c\}$ can be omitted as $p(x)=q(x)=0$.
		\begin{align*}
			p(a)q(b)+p(a)q(c)&+p(b)q(c)
			\geq q(a)p(b)+q(a)p(c)+q(b)p(c)
		\end{align*}
	
	Using the fact that $1=q(a)+q(b)+q(c)$ and $1=p(a)+p(b)+p(c)$, we have two possibilities of rewriting this inequality.
		\begin{align*}
			p(a)(1-q(a))+p(b)q(c)&\geq q(a)(1-p(a))+q(b)p(c)\\
			\iff p(a) +p(b)q(c)&\geq q(a)+q(b)p(c)
		\end{align*}
		\begin{align*}
			p(a)q(b) + (1-p(c))q(c) &\geq q(a)p(b) + (1-q(c))p(c)\\
			\iff p(a)q(b) + q(c) &\geq q(a)p(b) + p(c)
		\end{align*}
	
	Summing up these two inequalities results in the following inequality. 
		\begin{align*}
			p(a)(1+q(b)) +  q(c)(1+p(b)) &\geq q(a)(1+p(b)) + p(c)(1+q(b))\\
		    \iff (p(a)-p(c))(1+q(b)) &\geq (1+p(b))(q(a)-q(c))
		\end{align*}
	
		Our assumption that $p(a)<p(c)$ implies that the left-hand side of the inequality is smaller than $0$. On the other hand, we have $q(a)\geq q(c)$, so the right-hand side is non-negative. This is a contradiction, which proves that our assumptions on $p$ and $q$ are in conflict with \pc-strategyproofness. Hence, if $f(R,a)< f(R,c)$, then $f(R',a)< f(R',c)$.
	\end{proof}

    For our proofs, we also need insights on circumstances under which a lottery is \pc-inefficient. To this end, we analyze in the next lemma when voters prefer certain lotteries to each other. In order to succinctly formalize these results, we define the rank of an alternative $x$ in a preference relation $\succ_i$ as $r(\succ_i,x)=1+|\{y\in A\setminus \{x\}\colon y\succ_i x\}|$. In particular, if $r(\succ_i,x)<r(\succ_i,y)$, then $x\succ_i y$ due to the transitivity of preference relations.
	
	\begin{lemma}\label{lem:pcpref}
	Let $A=\{w,x,y,z\}$, let $p$ denote a lottery on $A$ with $p(x)>0$ and $p(y)>0$, and define $q$ as $q(x)=p(x)-\frac{\epsilon}{p(x)+p(z)}$, $q(y)=p(y)-\frac{\epsilon}{p(y)+p(z)}$, $q(z)=p(z)+\frac{\epsilon}{p(x)+p(z)}+\frac{\epsilon}{p(y)+p(z)}$, and $q(w)=p(w)$, where $\epsilon>0$ is sufficiently small so that $q(x)\geq 0$, $q(y)\geq 0$. Then, the following \pc-preferences hold:
	\begin{enumerate}
	    \item If $z\succ_i x$ and $z\succ_i y$, then $q\succ_i^\pc p$.
	    \item If $x\succ_i z$ and $y\succ_i z$, then $p \succ_i^\pc q$.
	    \item Assume that $x\succ_i z\succ_i y$ or $y\succ_i z\succ_i x$.
	    \begin{enumerate}
	        \item If $p(w)=0$ or $r(\succ_i, w)\in \{1,4\}$, then $p\succsim_i^\pc q$ and $q\succsim_i^\pc p$.
	        \item If $r(\succ_i, w)=2$ and $p(w)>0$, then $p\succ_i^\pc q$.
	        \item If $r(\succ_i, w)=3$ and $p(w)>0$, then $q\succ_i^\pc p$.
	    \end{enumerate}
	\end{enumerate}
	\end{lemma}
	\begin{proof}
	Let $p$ and $q$ be defined as in the lemma. Claims 1 and 2 follow immediately since a voter $i$ with $z\succ_i x$ and $z\succ_i y$ (resp. $x\succ_i z$ and $y\succ_i z$) strictly \sd-prefers $q$ to $p$ (resp. $p$ to $q$).
	
	For Claim 3, suppose that $x\succ_i z\succ_i y$; the case that $y\succ_i z\succ_i x$ is symmetric. The key insight for this case is that
	
	\begin{align*}
	    &p(x)\Big(q(z)+q(y)\Big)+p(z)q(y)\\
	    &=p(x)\Big(p(z)+p(y)+\frac{\epsilon}{p(x)+p(z)}\Big)+p(z)\Big(p(y)-\frac{\epsilon}{p(y)+p(z)}\Big)\\
	    &=p(x)\Big(p(z)+p(y)\Big)+p(z)p(y)+\frac{\epsilon p(x)}{p(x)+p(z)}-\frac{\epsilon p(z)}{p(y)+p(z)}\\
	    &=p(x)\Big(p(z)+p(y)\Big)+p(z)p(y)+\Big(\epsilon-\frac{\epsilon p(z)}{p(x)+p(z)}\Big)-\Big(\epsilon-\frac{\epsilon p(y)}{p(y)+p(z)}\Big)\\
	    &=p(x)\Big(p(z)+p(y)\Big)+p(z)p(y)-\frac{\epsilon p(z)}{p(x)+p(z)}+\frac{\epsilon p(y)}{p(y)+p(z)}+\frac{\epsilon p(y)-\epsilon p(y)}{p(x)+p(z)}\\
	    &=\Big(p(x)-\frac{\epsilon}{p(x)+p(z)}\Big)\Big(p(z)+p(y)\Big)+\Big(p(z)+\frac{\epsilon}{p(x)+p(z)}+\frac{\epsilon}{p(y)+p(z)}\Big)p(y)\\
	    &=q(x)\Big(p(z)+p(y)\Big) +q(z)p(y).
	\end{align*}
	
	As a consequence of this equation, voter $i$'s preference between $q$ and $p$ only depends on $p(w)$ and $r(\succ_i, w)$. First, if $p(w)=q(w)=0$, the \pc-comparison between $p$ and $q$ with respect to $\succ_i$ reduces exactly to the above equation. Hence, $p\succsim_i^\pc q$ and $q\succsim_i^\pc p$ if $p(w)=q(w)=0$. We therefore suppose that $p(w)=q(w)>0$ and proceed with a case distinction with respect to $r(\succ_i, w)$. In this analysis, we use $\Delta_{p\rightarrow q}^{xyz}=p(x)\Big(q(z)+q(y)\Big)+p(z)q(y)$ and $\Delta_{q\rightarrow p}^{xyz}=q(x)\Big(p(z)+p(y)\Big)+q(z)p(y)$ for a shorter notation.
    \begin{itemize}
        \item First, suppose that $r(\succ_i, w)=1$, i.e., $\succ_i=w,x,z,y$. Since $p(w)=q(w)$, $p(x)+p(y)+p(z)=1-p(w)$, and $q(x)+q(y)+q(z)=1-q(w)$, it is easy to verify that $p(w)\Big(q(x)+q(y)+q(z)\Big)+\Delta_{p\rightarrow q}^{xyz}= q(w)\Big(p(x)+p(y)+p(z)\Big)+\Delta_{q\rightarrow p}^{xyz}$. This implies that $p\succsim_i^\pc q$ and $q\succsim_i^\pc p$.
        \item As the second case, suppose that $r(\succ_i, w)=4$, i.e., $\succ_i=x,z,y,w$. For the same reason as in the first case it follows that $\Delta_{p\rightarrow q}^{xyz}+\Big(p(x)+p(y)+p(z)\Big)q(w)=\Delta_{q\rightarrow p}^{xyz}+\Big(q(x)+q(y)+q(z)\Big)p(w)$ and thus, $p\succsim_i^\pc q$ and $q\succsim_i^\pc p$. \item Next, suppose that $r(\succ_i, w)=2$, i.e., $\succ_i=x,w,z,y$. Then, $p\succ_i^\pc q$ as
        \begin{align*}
            \sum_{u,v\in A\colon u\succ_i v} p(u)q(v) &=p(x)q(w)+p(w)q(y)+p(w)q(z)+\Delta_{p\rightarrow q}^{xyz}\\
            &=p(x)p(w)+p(w)\Big(p(y)+p(z)+\frac{\epsilon}{p(x)+p(z)}\Big)+\Delta_{p\rightarrow q}^{xyz}\\
            &=\Big(p(x)+\frac{\epsilon}{p(x)+p(z)}\Big)p(w)+p(w)\Big(p(y)+p(z)\Big)+\Delta_{p\rightarrow q}^{xyz}\\
            &>q(x)p(w)+q(w)\Big(p(y)+p(z)\Big)+\Delta_{q\rightarrow p}^{xyz}\\
            &=\sum_{u,v\in A\colon u\succ_i v} q(u)p(v).
        \end{align*}
        \item As the last case, suppose that $r(\succ_i, w)=3$, i.e., $\succ_i = x,z,w,y$. In this case, a symmetric inequality as in the previous case proves that $q\succ_i^\pc p$.\qedhere
    \end{itemize}
	\end{proof}
	
	Note that \Cref{lem:pcpref} also holds for all lotteries $q$ and $p$ with $q(z)>0$ and $p(x)=q(x)+\frac{\epsilon'}{q(x)+q(z)}$, $p(y)=q(y)+\frac{\epsilon'}{q(y)+q(z)}$, $p(z)=q(z)-\frac{\epsilon'}{q(x)+q(z)}-\frac{\epsilon'}{q(y)+q(z)}$, and $p(w)=q(w)$ ($\epsilon'>0$ is again sufficiently small so that $p$ is a well-defined lottery on $\{w,x,y,z\}$). The reason for this is that $q(x)=p(x)-\frac{\epsilon}{p(x)+p(z)}$, $q(y)=p(y)-\frac{\epsilon}{p(y)+p(z)}$, and $q(z)=p(z)+\frac{\epsilon}{p(x)+p(z)}+\frac{\epsilon}{p(y)+p(z)}$ for $\epsilon=\epsilon'-\frac{\epsilon'^2}{(q(x)+q(z))\cdot(q(y)+q(z))}$. For instance, this follows for $p(x)$ from the following equation.  
	\begin{align*}
	    &p(x)-\frac{\epsilon}{p(x)+p(z)} =p(x)-\frac{\epsilon'-\frac{\epsilon'^2}{(q(x)+q(z))\cdot(q(y)+q(z))}}{q(x)+q(z)-\frac{\epsilon'}{q(y)+q(z)}}\\
	    &=p(x)-\frac{\epsilon'(q(x)+q(z))-\frac{\epsilon'^2}{q(y)+q(z)}}{(q(x)+q(z))\cdot(q(x)+q(z)-\frac{\epsilon'}{q(y)+q(z)})}=p(x)-\frac{\epsilon'}{q(x)+q(z)}=q(x)
	\end{align*}

	Finally, we are now ready to prove \Cref{lem:AWPSP}. We prove this statement separately for electorates with an odd number of voters and for those with an even number of voters.
	
	\begin{customlemma}{1a)}\label{lem:SPEffnodd}
	    Every \pc-efficient SDS that satisfies the absolute winner property is \pc-manipulable if $|N|\geq 3$ is odd and $m\geq 4$.
	\end{customlemma}
	\begin{proof}
	    Consider an arbitrary electorate $N\in\mathcal{F}(\mathbb{N})$ with an odd number of voters $n=|N|\geq 3$ and suppose there are $m\geq 4$ alternatives. We assume for contradiction that there is an \pc-efficient SDS $f$ that satisfies the absolute winner property and \pc-strategyproofness for $N$. 
	    In the sequel, we will focus on profiles on the alternatives $\{a,b,c,d\}$; all other alternatives are always ranked below these alternatives and therefore Pareto-dominated. Hence, \pc-efficiency entails for all subsequent profiles that $f(R,x)=0$ for all $x\in A\setminus \{a,b,c,d\}$, which means that these alternatives do not affect our further analysis. In slight abuse of notation, we therefore assume that $A=\{a,b,c,d\}$.
	    
	    We derive a contradiction by focusing on the profiles $R$ and $R'$ shown below. Specifically, our goal is to show that $f(R,a)=f(R,b)=f(R,c)=\frac{1}{3}$ and $f(R',a)=f(R',c)=f(R',d)=\frac{1}{3}$. This implies that voter $\frac{n+1}{2}$ can \pc-manipulate by switching from $R'$ to~$R$ as he even \sd-prefers $f(R)$ to $f(R')$, i.e., these claims result in a contradiction to \pc-strategyproofness.
	    
        \begin{profile}
		    $R$: & $\set{1}{\frac{n-1}{2}}$: $a,d,b,c$ & $\frac{n+1}{2}$: $b,c,d,a$
		    & $\set{\frac{n+3}{2}}{n}$: $c,a,d,b$\\
		    $R'$: & $\set{1}{\frac{n-1}{2}}$: $a,d,b,c$ & $\frac{n+1}{2}$: $b,d,c,a$
		    & $\set{\frac{n+3}{2}}{n}$: $c,a,d,b$
	    \end{profile}
	    
	    \textbf{Claim 1: $f(R,a)=f(R,b)=f(R,c)=\frac{1}{3}$}
	    
	    For proving this claim, our first goal is to establish that $f(R,c)>0$. Hence, assume for contradiction that this is not the case, i.e., $f(R,c)=0$. For deriving a contradiction to this assumption, we consider the profiles $R^1$ and $R^2$ shown below.
	    
	    \begin{profile}
	           $R^1$: & $\set{1}{\frac{n-1}{2}}$: $a,d,b,c$ & $\frac{n+1}{2}$: $b,c,d,a$
		    & $\set{\frac{n+3}{2}}{n-1}$: $c,a,d,b$ & $n$: $a,c,d,b$\\
		    $R^2$: & $\set{1}{\frac{n-1}{2}}$: $a,d,b,c$ & $\frac{n+1}{2}$: $c,b,d,a$
		    & $\set{\frac{n+3}{2}}{n}$: $c,a,d,b$\\
	    \end{profile}
	    
	    First, note that $a$ is top-ranked by more than half of the voters in $R^1$ and $c$ by more of half of the voters in $R^2$. Hence, the absolute winner property requires that $f(R^1, a)=f(R^2,c)=1$. On the other hand, $R^1$ is derived from $R$ by letting voter $n$ swap $a$ and $c$. Hence, \pc-strategyproofness, or more precisely \pcone-strategyproofness, from $R$ to $R^1$ implies that $f(R,c)\geq f(R, b)+f(R,d)$. Because we assume that $f(R,c)=0$, this means that $f(R,b)=f(R,d)=0$ and $f(R,a)=1$. On the other hand, the profile~$R^2$ is derived from $R$ by letting voter $\frac{n+1}{2}$ swap $b$ and $c$. Hence, \pc-strategyproofness requires that $f(R,b)\geq f(R,a)+f(R,d)$, which conflicts with $f(R,a)=1$. Thus, the initial assumption that $f(R,c)=0$ is incorrect, i.e., it holds that $f(R,c)>0$. 
	    
	    Departing from this insight, \pc-efficiency entails that $f(R,d)=0$. In more detail, \Cref{lem:pcpref} proves that every lottery $q$ with $q(d)>0$ and $q(c)>0$ is \pc-inefficient for $R$ because it is dominated by the lottery $p$ with $p(a)=q(a)+\frac{\epsilon}{q(a)+q(d)}$, $p(b)=q(b)+\frac{\epsilon}{q(b)+q(d)}$, $p(c)=q(c)$, and $p(d)=q(d)-\frac{\epsilon}{q(a)+q(d)}-\frac{\epsilon}{q(b)+q(d)}$. Indeed, Case 3a) of this lemma shows that all voters but $\frac{n+1}{2}$ are indifferent between $p$ and $q$, whereas Case~3b) implies that voter $\frac{n+1}{2}$ strictly prefers $p$ to $q$ (see also the text after \Cref{lem:pcpref}). Since we already know that $f(R,c)>0$, it follows therefore that $f(R,d)=0$.
	    
	    Next, note that the inequalities derived from \pc-strategy\-proofness on $R^1$ and $R^2$ remain valid, even if $f(R,c)>0$. Combined with the fact that $f(R,d)=0$, this means that $f(R,c)\geq f(R,b)\geq f(R,a)$. Hence, we prove Claim 1 by showing that $f(R,a)\geq f(R,c)$. Consider for this the profiles $\bar R^i$ for $i\in \{0,\dots, \frac{n-1}{2}\}$, which are defined as follows.
	    
	    \begin{profile}
	        $\bar R^i$: & $\set{1}{i}$: $a,d,b,c$ & $\set{i+1}{\frac{n-1}{2}}$: $b,a,d,c$
	        & $\frac{n+1}{2}$: $b,c,d,a$ & $\set{\frac{n+3}{2}}{n}$: $c,a,d,b$
	    \end{profile}
	    
	    First, note that $\bar R^{\frac{n-1}{2}}=R$ and that $f(\bar R^0, b)=1$ because $\frac{n+1}{2}$ voters report $b$ as their favorite alternative in this profile. Furthermore, \Cref{lem:pcpref} shows that $f(\bar R^i, d)=0$ for every profile $\bar R^i$ with $i<\frac{n-1}{2}$ because all lotteries $q$ with $q(d)>0$ fail \pc-efficiency for $\bar R^i$. Indeed, the lottery $p$ with $p(a)=q(a)+\frac{\epsilon}{q(a)+q(d)}$, $p(b)=q(b)+\frac{\epsilon}{q(b)+q(d)}$, $p(c)=q(c)$, and $p(d)=q(d)-\frac{\epsilon}{q(a)+q(d)}-\frac{\epsilon}{q(b)+q(d)}$ \pc-dominates $q$. Finally, by a repeated application of \Cref{lem:3alts}, we derive that $f(R,a)\geq f(R,c)$. To this end, consider a fixed index $i\in \{1,\dots, \frac{n-1}{2}\}$. If $f(\bar R^i,a)<f(\bar R^i,c)$, this lemma requires that $f(\bar R^{i-1},a)<f(\bar R^{i-1},c)$. Hence, if $f(R,a)<f(R,c)$, we can repeatedly apply this argument to derive that $f(\bar R^0, a)<f(\bar R^0,c)$. However, this contradicts the absolute winner property, and thus we must have that $f(R,a)\geq f(R,c)$. This proves Claim~1.\bigskip

	    \textbf{Claim 2: $f(R',a)=f(R',c)=f(R',d)=\frac{1}{3}$}
	    
	    As the second claim, we prove that $f$ assigns a probability of $\frac{1}{3}$ to $a$, $c$, and $d$ in $R'$. For this, we proceed analogously to Claim 1 and first show that $f(R',c)>0$. Assume for contradiction that $f(R',c)=0$ and consider the profiles $R^1$ and $R^2$ shown below. 
	    
	    \begin{profile}
		    $R^1$: & $\set{1}{\frac{n-1}{2}}$: $a,d,b,c$ & $\frac{n+1}{2}$: $b,d,c,a$
		    & $\set{\frac{n+3}{2}}{n-1}$: $c,a,d,b$ & $n$: $a,c,d,b$\\
		    $R^2$: & $\set{1}{\frac{n-1}{2}}$: $a,d,b,c$ & $\frac{n+1}{2}$: $c,d,b,a$
		    & $\set{\frac{n+3}{2}}{n}$: $c,a,d,b$
	    \end{profile}
	    
	    First, note that $f(R^1,a)=1$ and $f(R^2, c)=1$ because of the absolute winner property. Next, observe that $R^1$ is derived from $R'$ by letting voter $n$ swap $a$ and $c$. Hence, \pc-strategyproofness requires that $f(R',c)\geq f(R',b)+f(R',d)$. Since $f(R',c)=0$ by assumption, it follows that $f(R',b)=f(R',d)=0$ and $f(R',a)=1$. On the other hand, we derive $R^2$ from $R'$ by letting voter $\frac{n+1}{2}$ deviate. Hence, \pc-strategyproofness implies that $f(R',b)+f(R',d)\geq f(R',a)$, which contradicts $f(R',a)=1$. This shows that the initial assumption $f(R',c)=0$ is wrong, i.e., it must be that $f(R',c) > 0$.
	    
	    As the next step, we will infer from \pc-efficiency that $f(R', b)=0$. Assume that this is not the case, i.e., there is a \pc-efficient lottery $p$ with $p(b)>0$ and $p(c)>0$. Now, if $p(a)>0$, then $p$ is \pc-dominated by the lottery $q$ with $q(a)=p(a)-\frac{\epsilon}{p(a)+p(d)}$, $q(b)=p(b)-\frac{\epsilon}{p(b)+p(d)}$, $q(c)=p(c)$, and $q(d)=p(d)+\frac{\epsilon}{p(a)+p(d)}+\frac{\epsilon}{p(b)+p(d)}$ (where $\epsilon>0$ is so small that $q$ is a well-defined lottery). On the other hand, if $p(a)=0$, then $p$ is \pc-dominated by the lottery $q$ with $q(a)=p(a)$, $q(b)=p(b)-\frac{\epsilon}{p(b)+p(d)}$, $q(c)=p(c)-\frac{\epsilon}{p(c)+p(d)}$, 
 and $q(d)=p(d)+\frac{\epsilon}{p(b)+p(d)}+\frac{\epsilon}{p(c)+p(d)}$. Both of these claims are straightforward to verify with \Cref{lem:pcpref}. Since $p$ is \pc-inefficient in both cases, it follows that $f(R',b)=0$. 
	    
	    Just as for $R$, we can use the fact that $f(R',b)=0$ to simplify the inequalities derived from \pc-strategyproofness on $R^1$ and $R^2$. In particular, we infer that $f(R',c)\geq f(R',d)\geq f(R',a)$ from these observations. Hence, Claim 2 will follow by showing that $f(R',a)\geq f(R',c)$. Consider for this the profiles $\bar R^i$ for $i\in \{0,\dots, \frac{n-1}{2}\}$, which are defined as follows.
	    
	    \begin{profile}
		    $\bar R^{i}$: & $\set{1}{i}$: $a,d,b,c$ & $\set{i+1}{\frac{n-1}{2}}$: $d,a,b,c$ 
		    & $\frac{n+1}{2}$: $b,d,c,a$ & $\set{\frac{n+3}{2}}{n}$: $c,a,d,b$
			\end{profile}
			
		First, observe that $f(\bar R^i, b)=0$ for all $i\in \{0,\dots, \frac{n-3}{2}\}$ because of \pc-efficiency and \pc-strategyproofness. Indeed, assume for contradiction that this is not true, i.e., $f(\bar R^i, b)>0$ for some $i\in \{0,\dots, \frac{n-3}{2}\}$. First, we show that this assumption implies that $f(\bar R^i,a)=0$ and $f(\bar R^i, c)=0$ because of \pc-efficiency. For this, note that every lottery $p$ with $p(a)>0$ and $p(b)>0$ is \pc-dominated by the lottery $q$ with $q(a)=p(a)-\frac{\epsilon}{p(a)+p(d)}$, $q(b)=p(b)-\frac{\epsilon}{p(b)+p(d)}$, $q(c)=p(c)$, and $q(d)=p(d)+\frac{\epsilon}{p(a)+p(d)}+\frac{\epsilon}{p(b)+p(d)}$ in $\bar R^i$. Indeed, \Cref{lem:pcpref} shows that the voters $j\in\{i+1, \dots, \frac{n-1}{2}\}$ strictly \pc-prefer $q$ to $p$ and all other voters at least weakly \pc-prefer $q$ to $p$. Moreover, using again \Cref{lem:pcpref}, it is easy to see that every lottery $p$ with $p(a)=0$, $p(b)>0$, and $p(c)>0$ is \pc-dominated by the lottery $q$ with $q(a)=0$, $q(b)=p(b)-\frac{\epsilon}{p(b)+p(d)}$, $q(c)=p(c)-\frac{\epsilon}{p(c)+p(d)}$, 
 and $q(d)=p(d)+\frac{\epsilon}{p(b)+p(d)}+\frac{\epsilon}{p(c)+p(d)}$. Hence, if $f(\bar R^i,b)>0$, we derive that $f(\bar R^i,a)=0$ and $f(\bar R^i,c)=0$, which means that $\supp(f(\bar R^i))\subseteq \{b,d\}$. However, this entails that one of the voters $j\in \{\frac{n+1}{2},\dots,n\}$ can \pc-manipulate. Consider for this the subsequent preference profiles $\bar R^{i,j}$ for $j\in \{\frac{n+1}{2},\dots,n\}$ and $i\in \{0,\dots, \frac{n-3}{2}\}$.
		
		\begin{profile}
		    $\bar R^{i,j}$: & $\set{1}{i}$: $a,d,b,c$ & $\set{i+1}{\frac{n-1}{2}}$: $d,a,b,c$
		    & $\frac{n+1}{2}$: $b,d,c,a$\\ & $\set{\frac{n+3}{2}}{j}$: $c,a,d,b$ 
		    & $\set{j+1}{n}$: $d,a,b,c$
			\end{profile}
			
	    Note that $\bar R^{i,n} = \bar R^{i}$, and that $f(\bar R^{i,\frac{n+1}{2}},d)=1$ because more than half of the voters report $d$ as their favorite choice. On the other hand, we claim that for $j\in \{\frac{n+3}{2},\dots,n\}$, if $f(\bar R^{i,j},b)>0$, then $f(\bar R^{i, j-1}, b)>0$. Observe for this that the voter types in $\bar R^i$ and $\bar R^{i,j}$ coincide, and thus \pc-efficiency also requires that $f(\bar R^{i,j},a)=f(\bar R^{i,j},c)=0$ if $f(\bar R^{i,j},b)>0$. Moreover, \pc-strategyproofness requires that the deviating voter $j$ \pc-prefers $f(\bar R^{i,j})$ to $f(\bar R^{i,j-1})$. Now, if $f(\bar R^{i,j-1},b)=0$, deviating to $\bar R^{i,j-1}$ is a \pc-manipulation for voter $j$ because $\supp(f(\bar R^{i,j}))$ consists of his worst two alternatives. Hence, $f(\bar R^{i,j},b)>0$ implies that $f(\bar R^{i,j-1},b)>0$ and, by repeatedly applying this argument, we infer that if $f(\bar R^{i},b)>0$, then $f(\bar R^{i, \frac{n+1}{2}}, b)>0$. However, this contradicts $f(\bar R^{i, \frac{n+1}{2}}, d)=1$, so we have that $f(\bar R^i,b)=0$ for all $i\in \{0, \dots, \frac{n-3}{2}\}$.
	    
	    In particular, this argument proves also for the profile $\bar R^0$ that $f(\bar R^0,b)=0$. We show next that $f(\bar R^0,d)=1$. Consider for this the profile $\hat R$ derived from $\bar R^0$ by letting voter~$\frac{n+1}{2}$ swap $b$ and $d$. 
	    
	     \begin{profile}
		    $\hat R$: & $\set{1}{\frac{n-1}{2}}$: $d,a,b,c$ 
		    & $\frac{n+1}{2}$: $d,b,c,a$ & $\set{\frac{n+3}{2}}{n}$: $c,a,d,b$
			\end{profile}
	    
	    We have that $f(\hat R,d)=1$ because of the absolute winner property. Hence, \pc-strategyproofness requires that $f(\bar R^0,b)\geq f(\bar R^0, c)+f(\bar R^0,a)$. Since $f(\bar R^0,b)=0$, this means that $f(\bar R^0,c)=f(\bar R^0,a)=0$ and therefore $f(\bar R^0,d)=1$. Based on this observation, we can now use \Cref{lem:3alts} to derive that $f(R', a)\geq f(R',c)$. Consider for this an index $i\in \{1,\dots, \frac{n-1}{2}\}$ and suppose that $f(\bar R^{i-1},a)\geq f(\bar R^{i-1},c)$. The contraposition of \Cref{lem:3alts} shows that $f(\bar R^i,a)\geq f(\bar R^i,c)$ because the deviating voter~$i$ prefers $a$ to $d$ to $c$. Finally, since $f(\bar R^0,a)=f(\bar R^0,c)=0$, repeatedly applying the previous argument and noting that $R' = \bar R^{\frac{n-1}{2}}$, we obtain $f(R',a)\geq f(R',c)$.
	    This establishes Claim~2.
	\end{proof}
	
	Next, we turn to the proof for electorates with an even number of voters. Note that the proof follows a similar structure but requires more involved arguments because we cannot change the absolute winner by only modifying a single preference relation.
	
	\begin{customlemma}{1b)}\label{lem:SPEffneven}
		Every \pc-efficient SDS that satisfies the absolute winner property is \pc-manipulable if $|N|\geq 8$ is even and $m\geq 4$.
	\end{customlemma}
	\begin{proof}
	    Consider an arbitrary electorate $N\in \mathcal{F}(\mathbb{N})$ with $n=|N|\geq 8$ even and assume for contradiction that there is an SDS $f$ for $m\geq 4$ alternatives that satisfies \pc-efficiency, \pc-strategyproofness, and the absolute winner property on $N$. We focus on the case $m=4$ because we can generalize the constructions to larger values of $m$ by simply ranking the additional alternatives at the bottom. Then, \pc-efficiency requires that these alternatives obtain probability~$0$ and they therefore do not affect our analysis. We derive a contradiction by analyzing the following two profiles.

	    \begin{profile}
		    $R$: & $\set{1}{\frac{n}{2}-1}$: $a,d,b,c$ & $\{\frac{n}{2}, \frac{n}{2}+1\}$: $b,c,d,a$
		    & $\set{\frac{n}{2}+2}{n}$: $c,a,d,b$\\
		    $R'$: & $\set{1}{\frac{n}{2}-1}$: $a,d,b,c$ & $\{\frac{n}{2}, \frac{n}{2}+1\}$: $b,d,c,a$
		    & $\set{\frac{n}{2}+2}{n}$: $c,a,d,b$
	    \end{profile}
		
		In more detail, we show in Claims~1 and 2 that $f(R,a)=f(R,b)=f(R,c)=\frac{1}{3}$ and $f(R',a)=f(R',c)=f(R',d)=\frac{1}{3}$. These two claims are in conflict with \pc-strategyproofness, as the following analysis shows. Let $R''$ denote the profile ``between'' $R$ and $R'$ in which voter $\frac{n}{2}$ reports $b,d,c,a$ and voter $\frac{n}{2}+1$ reports $b,c,d,a$. 
		
		\begin{profile}
		    $R''$: & $\set{1}{\frac{n}{2}-1}$: $a,d,b,c$ & $\frac{n}{2}$: $b,d,c,a$ &  $\frac{n}{2}+1$: $b,c,d,a$ & $\set{\frac{n}{2}+2}{n}$: $c,a,d,b$
		\end{profile}
		
		Moreover, let $p=f(R)$, $q=f(R')$, and $r=f(R'')$ denote the outcome of $f$ in these profiles. \pc-strategyproofness from $R'$ to $R''$ results in the following inequality because $q(a) = q(c) = q(d)=\frac{1}{3}$.
		\begin{align*}
		    & q(b)\Big(r(d)+r(c)+r(a)\Big)+q(d)\Big(r(c)+r(a)\Big)+q(c)r(a)\\
		    & \geq  r(b)\Big(q(d)+q(c)+q(a)\Big)+r(d)\Big(q(c)+q(a)\Big)+r(c)q(a)\\
			\iff &\frac{1}{3}r(c)+\frac{2}{3}r(a)\geq r(b) + \frac{2}{3} r(d) + \frac{1}{3} r(c)
			\iff  r(a)\geq \frac{3}{2} r(b) + r(d)
		\end{align*}
		
		Moreover, we can also use \pc-strategyproofness from $R''$ to $R$ and the fact that $p(a) = p(b) = p(c) = \frac{1}{3}$ to infer the following inequality. 
		\begin{align*}
		& r(b)\Big(p(d)+p(c)+p(a)\Big)+r(d)\Big(p(c)+p(a)\Big)+r(c)p(a)\\
		    &\geq  p(b)\Big(r(d)+r(c)+r(a)\Big)+p(d)\Big(r(c)+r(a)\Big)+p(c)r(a)\\
	     \iff &\frac{2}{3}r(b) + \frac{2}{3}r(d) + \frac{1}{3}r(c) \geq \frac{1}{3}r(d) + \frac{1}{3}r(c) + \frac{2}{3} r(a)
			\iff  r(b) + \frac{1}{2} r(d) \geq  r(a)
		\end{align*}
		
		Combining these two inequalities entails that $r(b)+\frac{1}{2}r(d)\geq \frac{3}{2}r(b) + r(d)$, which is true only if $r(b) = r(d)=0$. Moreover, the second inequality upper bounds $r(a)$ and thus $r(a)=0$. This means that $f(R'',c)=r(c)=1$. However, $c$ is the worst alternative of the voters $\seti{1}{\frac{n}{2}-1}$ and \pc-strategyproofness hence requires that these voters cannot affect the outcome by misreporting their preferences. On the other hand, if we let these voters one after another change their preference relation to $b,d,a,c$, we arrive at a profile in which $b$ is top-ranked by more than half of the voters. Hence, the absolute winner property requires that $b$ is chosen with probability $1$, which is in conflict with the observation that these voters are not able to affect the outcome. This is the desired contradiction.
		Hence, to complete the proof of Lemma~\ref{lem:SPEffneven}, it remains to show the claims for $f(R)$ and $f(R')$.\bigskip
		
		\textbf{Claim 1: $f(R,a)=f(R,b)=f(R,c)=\frac{1}{3}$}
		
		Just as in the case of odd $n$, our first goal is to prove that $f(R,d)=0$. As the first step in proving this statement, we show that $f(R,a)<1$. Hence, assume for contradiction that $f(R,a)=1$, which means that the least preferred lottery of voters $\frac{n}{2}$ and $\frac{n}{2}+1$ is chosen. Moreover, if both of these voters swap $b$ and $c$, $c$ is top-ranked by more than half of the voters, so the absolute winner property requires $c$ to receive a probability of $1$. This is, however, in conflict with \pc-strategyproofness, which requires that these voters cannot affect the outcome. Hence, the assumption that $f(R,a)=1$ must have been wrong, i.e., $f(R,a)<1$.
		
		Based on this insight, we show by contradiction that $f(R,c)>0$, i.e., suppose that $f(R,c)=0$. Next, consider the profiles $R^1$ and $R^2$ shown below.
		
		\begin{profilenarrow}
			$R^1$: & $\set{1}{\frac{n}{2}-1}$: $a,d,b,c$ & $\{\frac{n}{2}, \frac{n}{2}+1\}$: $b,c,d,a$ 
		    & $\set{\frac{n}{2}+2}{n-1}$: $c,a,d,b$ & $n$: $a,c,d,b$\\
		    $R^2$: & $\set{1}{\frac{n}{2}-1}$: $a,d,b,c$ & $\{\frac{n}{2}, \frac{n}{2}+1\}$: $b,c,d,a$ & $\set{\frac{n}{2}+2}{n-2}$: $c,a,d,b$ & $\{n-1, n\}$: $a,c,d,b$
		\end{profilenarrow}
		
		Alternative $a$ is top-ranked by $\frac{n}{2}+1$ voters in $R^2$, which means that $f(R^2,a)=1$ because of the absolute winner property. Now, using \pc-strategyproofness (or more precisely \pcone-strategyproofness) from $R^1$ to $R^2$, we derive that $f(R^1,c)\geq f(R^1,d)+f(R^1,b)$. In particular, this inequality requires that $f(R^1,a)=1$ if $f(R^1,c)=0$. However, in that case, voter $n-1$ can \pc-manipulate by deviating from $R$ to $R^1$: since $f(R,c)=0$ and $f(R,a)<1$, it follows that $f(R^1)\succ_{n-1}^\sd f(R)$ and therefore also $f(R^1)\succ_{n-1}^\pc f(R)$. Hence, it must hold that $f(R^1,c)>0$. Next, we use \pc-strategyproofness from $R$ to $R^1$ and vice versa to derive the following two inequalities, where $p=f(R)$ and $q=f(R^1)$.
		\begin{align*}
		   & p(c)\Big(q(a)+q(d)+q(b)\Big)+p(a)\Big(q(d)+q(b)\Big)+p(d)q(b)\\
		   &\qquad\qquad\qquad\geq q(c)\Big(p(a)+p(d)+p(b)\Big)+q(a)\Big(p(d)+p(b)\Big)+q(d)p(b)\\
		   & q(a)\Big(p(c)+p(d)+p(b)\Big)+q(c)\Big(p(d)+p(b)\Big)+q(d)p(b)\\
		   &\qquad\qquad\qquad\geq p(a)\Big(q(c)+q(d)+q(b)\Big)+p(c)\Big(q(d)+q(b)\Big)+p(d)q(b)
		\end{align*}
		
		Adding these two inequalities and cancelling common terms yields $p(c)q(a)\geq q(c)p(a)$. Since $p(c)=0$ by assumption and $q(c)>0$ because of our previous analysis, this inequality can only be true if $p(a)=0$. Using the facts that $p(c)=p(a)=0$ and $q(c)\geq q(b)+q(d)$, we can therefore vastly simplify the first inequality. 
		\begin{align*}
		   p(d)q(b)&\geq q(c)+q(a)+q(d)p(b)\\
		   &\geq q(d)+q(b)+q(a)+q(d)p(b)
		\end{align*}
		
		It is easy to see that this inequality can only be true if $f(R,d)=p(d)=1$. We now derive a contradiction to this insight. First, from $R$, let voters $\frac{n}{2}$ and $\frac{n}{2}+1$ make $d$ into their favorite alternative. This leads to the profile $R^3$ (see below) and \pc-strategyproofness (one step at a time) requires that $f(R^3,d)=1$. Moreover, note that $d$ Pareto-dominates $b$ in $R^3$. Next, we let voters $n-1$ and $n$ swap $a$ and $c$ to obtain the profile $R^4$. \pc-efficiency requires that $f(R^4,b)=0$ as this alternative is still Pareto-dominated, and \pc-strategyproofness requires in turn that $f(R^4,d)=1$ as any other lottery with support $\{a,c,d\}$ yields a \pc-manipulation for voters $n-1$ and $n$. However, this contradicts the absolute winner property as $\frac{n}{2}+1$ voters report $a$ as their favorite alternative in $R^4$, so it cannot be the case that $f(R,d)=1$. Thus, no feasible outcome for $f(R)$ remains, which demonstrates that the assumption that $f(R,c)=0$ is wrong. That is, we must have $f(R,c) > 0$.
		
		\begin{profilenarrow}
		    $R^3$: & $[1\dots \frac{n}{2}-1]$: $a,d,b,c$ & $\{\frac{n}{2}, \frac{n}{2}+1\}$: $d,b,c,a$  & $[\frac{n}{2}+2\dots n]$: $c,a,d,b$\\
		    $R^4$: & $[1\dots \frac{n}{2}-1]$: $a,d,b,c$ & $\{\frac{n}{2}, \frac{n}{2}+1\}$: $d,b,c,a$  & $[\frac{n}{2}+2\dots n-2]$: $c,a,d,b$ & $\{n-1, n\}$: $a,c,d,b$
		\end{profilenarrow}
		
		As the next step, we can use an analogous argument as in \Cref{lem:SPEffnodd} to derive that $f(R,d)=0$ due to \pc-efficiency. Indeed, this follows immediately since the profile $R$ here and in the proof of \Cref{lem:SPEffnodd} consists of the same voter types and $f(R, c)>0$.
		
		Based on this insight, we show now that $f(R,c)\geq f(R,b)\geq f(R,a)\geq f(R,c)$, which implies that $f(R,a)=f(R,b)=f(R,c)=\frac{1}{3}$. For the first inequality, consider the profiles $R^5$ and $R^6$. 
		
		\begin{profilenarrow}
		       $R^5$: & $[1\dots \frac{n}{2}-1]$: $a,d,b,c$ & $\{\frac{n}{2}, \frac{n}{2}+1\}$: $b,c,d,a$
		    & $[\frac{n}{2}+2\dots n-1]$: $c,a,d,b$ & $n$: $a,c,b,d$\\
		    $R^6$: & $[1\dots \frac{n}{2}-1]$: $a,d,b,c$ & $\{\frac{n}{2}, \frac{n}{2}+1\}$: $b,c,d,a$
		    & $[\frac{n}{2}+2\dots n-2]$: $c,a,d,b$ & $\{n-1, n\}$: $a,c,b,d$
		\end{profilenarrow}
		
		\Cref{lem:pcpref} and \pc-efficiency imply that $f(R^5,d)=0$ as we can otherwise find a lottery that $\pc$-dominates $f(R^5)$ by redistributing probability from $d$ to $a$ and $b$. Moreover, the absolute winner property shows that $f(R^6,a)=1$. In particular, $f(R,d)=f(R^5,d)=f(R^6,d)=0$ and we can thus apply \Cref{lem:3alts} twice to derive that $f(R^5, c)\geq f(R^5, b)$ and $f(R,c)\geq f(R,b)$ since $f(R^6,c)= f(R^6,b)=0$.
		
		Next, we show that $f(R,b)\geq f(R,a)$. Consider for this the profiles $R^7$ and $R^8$ derived from $R$ by replacing the preference relations of voters $\frac{n}{2}$ and $\frac{n}{2}+1$ sequentially with $c,a,b,d$.
		
		\begin{profilenarrow}
		       $R^7$: & $[1\dots \frac{n}{2}-1]$: $a,d,b,c$ & $\frac{n}{2}$: $c,a,b,d$ &  $\frac{n}{2}+1$: $b,c,d,a$ & $[\frac{n}{2}+2\dots n]$: $c,a,d,b$\\
		    $R^8$: & $[1\dots \frac{n}{2}-1]$: $a,d,b,c$ & $\{\frac{n}{2}, \frac{n}{2}+1\}$: $c,a,b,d$ & $[\frac{n}{2}+2\dots n]$: $c,a,d,b$
		\end{profilenarrow}
		
		First, observe that $f(R^8,c)=1$ as all voters $\seti{\frac{n}{2}}{n}$ report $c$ as their best choice. Furthermore, \Cref{lem:pcpref} and \pc-efficiency show that every lottery $q$ with $q(d)>0$ is \pc-inefficient in $R^7$ because we can find a \pc-improvement by redistributing the probability of $d$ to $a$ and $b$.  Hence, $f(R,d)=f(R^7,d)=f(R^8,d)=0$ and applying \Cref{lem:3alts} twice then shows that $f(R,b)\geq f(R,a)$ because $f(R^8,b)= f(R^8,a)=0$. 
		
		Finally, we prove that $f(R,a)\geq f(R,c)$. Consider for this the profiles $\bar R^{k}$ for $k\in \{0, \dots, \frac{n}{2}-1\}$ defined as follows.
		
		\begin{profilenarrow}
	    $\bar R^{k}$: & $[1\dots k]$: $a,d,b,c$ & $[k+1\dots \frac{n}{2}-1]$: $b,a,d,c$ 
	    & $\{\frac{n}{2}, \frac{n}{2}+1\}$: $b,c,d,a$  & $[\frac{n}{2}+2\dots n]$: $c,a,d,b$
		\end{profilenarrow}
		
	 Note that $R=\bar R^{\frac{n}{2}-1}$ and that $f(\bar R^{0},b)=1$ because of the absolute winner property. Moreover, $f(\bar R^{k},d)=0$ for all $k\in \{0,\dots, \frac{n}{2}-2\}$ because of \pc-efficiency: once again, \Cref{lem:pcpref} shows that any lottery $q$ with $q(d)>0$ is \pc-dominated by the lottery $p$ with $p(a)=q(a)+\frac{\epsilon}{q(d)+q(a)}$, $p(b)=q(b)+\frac{\epsilon}{q(d)+q(b)}$, $p(c)=q(c)$, and $p(d)=q(d)-\frac{\epsilon}{q(d)+q(a)}-\frac{\epsilon}{q(d)+q(b)}$
	 (where $\epsilon>0$ is so small that $p(d)\geq 0$). 
	 Hence, $d$ receives probability~$0$ for all of these profiles. 
	 We also know that $f(R,d) = 0$, so $f(\bar R^{k},d)=0$ for all $k\in \{0,\dots, \frac{n}{2}-1\}$.
	 By inductively applying \Cref{lem:3alts}, we derive that $f(R,a)\geq f(R,c)$ because $f(\bar R^{0},a)= f(\bar R^{0},c)=0$.
	 This completes the proof of Claim~1. \bigskip
	 
	 \textbf{Claim 2: $f(R',a)=f(R',c)=f(R',d)=\frac{1}{3}$}
	 
	 For proving this claim, we show as the first step that $f(R',b)=0$. Note for this that an analogous argument as in Claim 1 proves that $f(R',c)>0$. Based on this insight, an analogous argument as in the proof of \Cref{lem:SPEffnodd} shows that $f(R',b)=0$ because of  \pc-efficiency and \Cref{lem:pcpref}. Indeed, this is straightforward as the profile $R'$ here and in the proof of \Cref{lem:SPEffnodd} consists of the same voter types and thus, the same lotteries are \pc-efficient.
	 
	 Using this observation, we show next that $f(R',c)\geq f(R',d)\geq f(R',a)\geq f(R',c)$, which implies that all three alternatives receive a probability of $\frac{1}{3}$. First, we prove that $f(R',c)\geq f(R',d)$ by considering the profiles $R^1$ and $R^2$.
	 
	 \begin{profilenarrow}
		    $R^1$: & $[1\dots \frac{n}{2}-1]$: $a,d,b,c$ & $\{\frac{n}{2}, \frac{n}{2}+1\}$: $b,d,c,a$ 
		    & $[\frac{n}{2}+2\dots n-1]$: $c,a,d,b$ & $n$: $a,d,c,b$\\
		    $R^2$: & $[1\dots \frac{n}{2}-1]$: $a,d,b,c$ & $\{\frac{n}{2}, \frac{n}{2}+1\}$: $b,d,c,a$ 
		    & $[\frac{n}{2}+2\dots n-2]$: $c,a,d,b$ & $\{n-1, n\}$: $a,d,c,b$
		\end{profilenarrow}
	 
	 Note that $f(R^2,a)=1$ because more than half of the voters rank $a$ top in $R^2$. Consequently, \pcone-strategyproofness entails that $f(R^1,c)\geq f(R^1,d)+f(R^1,b)$. Hence, if $f(R^1,c)=0$, then $f(R^1,a)=1$ and another application of \pcone-strategyproofness shows that $f(R',c)\geq f(R',d)$ because $f(R',b)=0$. On the other hand, if $f(R^1,c)>0$, \pc-efficiency requires that $f(R^1,b)=0$. In more detail, every lottery with $q(a)>0$, $q(b)>0$, and $q(c)>0$ is \pc-dominated in $R^1$ by the lottery $p$ with $p(a)=q(a)-\frac{\epsilon}{q(a)+q(d)}$, $p(b)=q(b)-\frac{\epsilon}{q(b)+q(d)}$, $p(c)=q(c)$, and $p(d)=q(d)+\frac{\epsilon}{q(a)+q(d)}+\frac{\epsilon}{q(b)+q(d)}$, whereas every lottery $q$ with $q(a)=0$, $q(b)>0$, and $q(c)>0$ is \pc-dominated in $R^1$ by the lottery $p$ with $p(a)=q(a)$, $p(b)=q(b)-\frac{\epsilon}{q(b)+q(d)}$, $p(c)=q(c)-\frac{\epsilon}{q(c)+q(d)}$, and $p(d)=q(d)+\frac{\epsilon}{q(b)+q(d)}+\frac{\epsilon}{q(c)+q(d)}$ (see \Cref{lem:pcpref}).
	 Hence, we have $f(R',b)=f(R^1,b)=f(R^2,b)=0$ and a repeated application of \Cref{lem:3alts} shows that $f(R',c)\geq f(R',d)$ because $f(R^2,c)= f(R^2,d)=0$. 
	 
	 Next, we derive that $f(R',d)\geq f(R',a)$.
	 To this end, consider the profiles $R^3$ and $R^4$ derived from $R'$ by replacing the preference relations of voters $\frac{n}{2}$ and $\frac{n}{2}+1$ with $c,a,d,b$. 
	 
	  \begin{profilenarrow}
		    $R^3$: & $[1\dots \frac{n}{2}-1]$: $a,d,b,c$ & $\frac{n}{2}$: $c,a,d,b$
		    &  $\frac{n}{2}+1$: $b,d,c,a$ & $[\frac{n}{2}+2\dots n]$: $c,a,d,b$ \\
		    $R^4$: & $[1\dots \frac{n}{2}-1]$: $a,d,b,c$ & $\{\frac{n}{2}, \frac{n}{2}+1\}$: $c,a,d,b$ & $[\frac{n}{2}+2\dots n]$: $c,a,d,b$ 
		\end{profilenarrow}
	 
	 It follows from the absolute winner property that $f(R^4,c)=1$ as more than half of the voters report $c$ as their best alternative. Moreover, we can use the same construction as for $R$ (in Claim 1) to derive that $f(R^3,c)>0$. Indeed, voters $n$ and $n-1$ have the same preference relations in $R$ and $R^3$ and they also can make $a$ into the absolute winner by swapping $a$ and $c$ in $R^3$. Analogously to $R'$, it now follows from \pc-efficiency and \Cref{lem:pcpref} that $f(R^3,b)=0$. Finally, a repeated application of \Cref{lem:3alts} shows that $f(R',d)\geq f(R',a)$ because $f(R',b)=f(R^3,b)=f(R^4,b)=0$ and $f(R^4,d)=f(R^4,a)=0$. 
		
	It remains to show that $f(R',a)\geq f(R',c)$. Consider the sequence of profiles $\hat R^{k}$ for $k\in \{0,\dots, \frac{n}{2}-1\}$ and note that $R'=\hat R^{\frac{n}{2}-1}$, which means that $f(\hat R^{\frac{n}{2}-1}, b)=0$.
		
		\begin{profilenarrow}
	    $\hat R^{k}$: & $[1\dots k]$: $a,d,b,c$ & $[k+1\dots \frac{n}{2}-1]$: $d,a,b,c$
	    & $\{\frac{n}{2}, \frac{n}{2}+1\}$: $b,d,c,a$ & $[\frac{n}{2}+2\dots n]$: $c,a,d,b$
		\end{profilenarrow}
		
		First, we show that $f(\hat R^k, b)=0$ for all other profiles $\hat R^{k}$ with $k\in \{0,\dots,\frac{n}{2}-2\}$. To this end, note that \Cref{lem:pcpref} shows that either $f(\hat R^{k},b)=0$ or $f(\hat R^{k},a)=0$ because otherwise, we can find a \pc-improvement by redistributing probability from $a$ and $b$ to~$d$. Moreover, if $f(\hat R^{k},a)=0$, then \Cref{lem:pcpref} entails that $f(\hat R^{k},b)=0$ or $f(\hat R^{k},c)=0$, because otherwise the probability from $b$ and $c$ can be redistributed to $d$. Now, assume for contradiction that $f(\hat R^{k},b)>0$ for a fixed $k$ and hence $f(\hat R^{k},a)=f(\hat R^{k},c)=0$. We proceed with a case distinction on $k$. First, suppose that $\frac{n}{2}-2\geq k\geq 2$, which means that at least two voters in $\set{1}{k}$ top-rank $a$. For this case, we consider the profiles $\hat R^{k,j}$ for $j\in \{\frac{n}{2}+1,\dots, n\}$.
		
		\begin{profile}
	    $\hat R^{k,j}$: & $[1\dots k]$: $a,d,b,c$ & $[k+1\dots \frac{n}{2}-1]$: $d,a,b,c$  & $\{\frac{n}{2}, \frac{n}{2}+1\}$: $b,d,c,a$ \\& $[\frac{n}{2}+2\dots j]$: $a,d,b,c$ 
	    &$[j+1\dots n]$: $c,a,d,b$
		\end{profile}
		
		It holds by definition that $\hat R^{k}=\hat R^{k,\frac{n}{2}+1}$. Moreover, analogous to $\hat R^k$, \pc-efficiency requires that either $f(\hat R^{k,j},a)+f(\hat R^{k,j},c)=0$ or $f(\hat R^{k,j},b)=0$ for every $j\in \{\frac{n}{2}+1,\dots, n\}$. This implies for every $j$ that if $f(\hat R^{k,j},b)>0$, then $f(\hat R^{k,j+1},b)>0$. In more detail, $f(\hat R^{k,j},b)>0$ requires that $f(\hat R^{k,j},b)+f(\hat R^{k,j},d)=1$ because of \pc-efficiency. Hence, if $f(\hat R^{k,j+1}, b)=0$, voter $j+1$ can \pc-manipulate by deviating from $\hat R^{k,j}$ to $\hat R^{k,j+1}$ since $b$ and $d$ are his worst two alternatives in $\hat R^{k,j}$. By repeatedly applying this argument, it follows that if $f(\hat R^k,b)>0$, then $f(\hat R^{k,n},b)>0$. However, this is in conflict with the absolute winner property as $\frac{n}{2}-1+k$ voters top-rank $a$ in $\hat R^{k,n}$. This proves that $f(\hat R^k,b)=0$ for all $k\in \{2,\dots, \frac{n}{2}-2\}$. 
		
		Note that the argument above fails if $k\leq 1$ as no more than $\frac{n}{2}$ voters top-rank $a$ in $\hat R^{1,n}$ and $\hat R^{0,n}$. Hence, we investigate the case $k\leq 1$ separately and consider for this the profiles $\tilde R^{k,j}$ for $j\in \{\frac{n}{2}+1,\dots, n\}$.
		
		\begin{profile}
			$\tilde R^{k,j}$: & $[1\dots k]$: $a,d,b,c$ & $[k+1\dots \frac{n}{2}-1]$: $d,a,b,c$ 
			& $\{\frac{n}{2}, \frac{n}{2}+1\}$: $b,d,c,a$ \\& $[\frac{n}{2}+2\dots j]$: $d,a,b,c$ 
			&$[j+1\dots n]$: $c,a,d,b$
		\end{profile}
		
		It holds by definition that $\hat R^{k}=\tilde R^{k,\frac{n}{2}+1}$.
		Note that the profiles $\tilde R^{k,j}$ consist of the same preference relations as $\hat R^{k}$ and hence, \pc-efficiency once again requires that $f(\tilde R^{k,j},b)=0$ or $f(\tilde R^{k,j},a)=f(\tilde R^{k,j},c)=0$. 
		(When $j = n$, even though the preference relation $c,a,d,b$ is not present, a similar argument still holds.)
		Moreover, if $f(\tilde R^{k,j},b)>0$, then $f(\tilde R^{k,j+1},b)>0$. The reason for this is that if $f(\tilde R^{k,j},b)>0$, then $f(\tilde R^{k,j},b)+f(\tilde R^{k,j},d)=1$. Hence, if $f(\tilde R^{k,j+1},b)=0$, voter $j+1$ can \pc-manipulate by deviating from $\tilde R^{k,j}$ to $\tilde R^{k,j+1}$ as $b$ and $d$ are his least preferred alternatives in $\tilde R^{k,j}$. By repeatedly applying this argument, we derive that if $f(\hat R^k, b)>0$, then $f(\tilde R^{k,n},b)>0$. However, this contradicts the absolute winner property as at least $(\frac{n}{2}-1-k)+(\frac{n}{2}-1)\geq n-3 > \frac{n}{2}$ voters top-rank $d$ in $\tilde R^{k,n}$. (Here we use the assumption that $n\ge 8$.) Hence, we also have that $f(\hat R^{k},b)=0$ if $k\leq 1$.
		
		As the last point, we prove that $f(\hat R^0,d)=1$. Then, it follows from repeated application of \Cref{lem:3alts} that $f(R',a)\geq f(R',c)$ because $f(\hat R^0,a)=f(\hat R^0,c)=0$. Thus, consider the profiles $R^5$ and $R^6$ derived from $\hat R^{0}$ by sequentially replacing the preference of voter $\frac{n}{2}$ and $\frac{n}{2}+1$ with $d,a,b,c$. 
		
		\begin{profilenarrow}
		$R^5$: & $[1\dots \frac{n}{2}-1]$: $d,a,b,c$ & $\frac{n}{2}$: $d,a,b,c$
		& $\frac{n}{2}+1$: $b,d,c,a$ & $[\frac{n}{2}+2\dots n]$: $c,a,d,b$\\
	    $R^6$: & $[1\dots \frac{n}{2}-1]$: $d,a,b,c$ & $\{\frac{n}{2}, \frac{n}{2}+1\}$: $d,a,b,c$ & $[\frac{n}{2}+2\dots n]$: $c,a,d,b$
		\end{profilenarrow}
		
		Note that an absolute majority top-ranks $d$ in $R^6$, which means that $f(R^6,d)=1$. Hence, \pcone-strategyproofness entails for $R^5$ that $f(R^5,b)\geq f(R^5,c)+f(R^5,a)$. If $f(R^5,b)=0$, we derive then that $f(R^5,d)=1$, and an application of \pcone-strategyproofness between $\hat R^0$ and $R^5$ shows that $f(\hat R^0, b)\geq f(\hat R^0 ,c)+f(\hat R^0, a)$. Since we already established that $f(\hat R^0, b)=0$, this proves that $f(\hat R^0,d)=1$. On the other hand, if $f(R^5,b)>0$, \pc-efficiency requires that $f(R^5,a)=f(R^5,c)=0$ (see \Cref{lem:pcpref}). However, then voter $n$ can \pc-manipulate in $R^5$ by reporting $d$ as his favorite alternative. Thereafter, $d$ must be chosen with probability $1$ because it is top-ranked by $\frac{n}{2}+1$ voters. However, voter $n$ \pc-prefers this lottery to $f(R^5)$ if $f(R^5,b)>0$, which contradicts the \pc-strategyproofness of $f$. Hence, it must indeed hold that $f(R^5,b)=0$ and therefore also $f(R^5,d)=1$ and $f(\hat R^0, d)=1$. Finally, as mentioned earlier in this paragraph, a repeated application of \Cref{lem:3alts} shows now that $f(R',a)\geq f(R',c)$. Therefore, we have that $f(R',c)\geq f(R',d)\geq f(R',a)\geq f(R',c)$ and $f(R',b)=0$, which implies that $f(R',a)=f(R',c)=f(R',d)=\frac{1}{3}$. 
	\end{proof}
	
	Since Lemmas~\ref{lem:SPEffnodd} and \ref{lem:SPEffneven} combined are clearly equivalent to \Cref{lem:AWPSP}, this concludes the first step for the proof of \Cref{thm:EffSP}. Next, we show that every anonymous, neutral, \pc-efficient, and \pc-strategyproof SDS also satisfies the absolute winner property. This insight together with \Cref{lem:AWPSP} immediately implies \Cref{thm:EffSP}. Note that the subsequent lemma is slightly stronger than required: we show that the implication also holds for $m\geq 3$ and all electorates.
    
	\begin{lemma}\label{lem:abswinner}
	Assume that $m\geq 3$. Every SDS that satisfies \pc-efficiency, \pc-strategyproofness, neutrality, and anonymity also satisfies the absolute winner property. 
	\end{lemma}
	\begin{proof}
	    Let $f$ denote an SDS that satisfies anonymity, neutrality, \pc-efficiency, and \pc-strategyproofness for $m\geq 3$ alternatives. First, note that for electorates $N$ with $n=|N|\leq 2$, the absolute winner property requires that an alternative is chosen with probability $1$ if it is top-ranked by all voters. This is clearly implied by \pc-efficiency and we thus focus on the case that $n\geq 3$. Since the construction in the main body (in the proof of \Cref{thm:EffSP}) works for every even $n\geq 4$, we only need to show the lemma for the case that the number of voters is odd. Hence, consider an electorate $N$ with an odd number of voters $n\geq 3$. Moreover, we focus on the case that there are $m=3$ alternatives because we can generalize all steps by simply adding additional alternatives at the bottom of all preference relations. \pc-efficiency then requires that these alternatives get probability~$0$ and they thus do not affect our analysis.  
		
		We start our analysis by considering the profiles $R^1$ and $R^2$ described below.
		
		\begin{profile}
	    $R^1$: & $[1\dots \frac{n-1}{2}]$: $b,a,c$ & $\frac{n+1}{2}$: $a,b,c$
	    & $[\frac{n+3}{2}\dots n]$: $c,a,b$ \\
	    $R^2$: & $[1\dots \frac{n-1}{2}]$: $b,a,c$ & $\frac{n+1}{2}$: $a,c,b$
	    & $[\frac{n+3}{2}\dots n]$: $c,a,b$ 
		\end{profile}
		
		First, note that anonymity and neutrality imply that $f(R^1,a)=f(R^2,a)$, $f(R^1,b)=f(R^2,c)$, and $f(R^1,c)=f(R^2,b)$.
		Furthermore, \pc-efficiency shows that $f(R^1,b)=f(R^2,c)=0$ or $f(R^1,c)=f(R^2,b)=0$. Subsequently, we show that $f(R^1,b)=f(R^1,c)=0$ must be true, which means that $f(R^1,a)=1$.
		
		Assume for contradiction that $f(R^1,c)=f(R^2,b)>0$. Then, our previous observation implies that $f(R^2,c)=f(R^1,b)=0$. However, this means that voter $\frac{n+1}{2}$ can manipulate by deviating from $R^1$ to $R^2$ because he \pc-prefers $f(R^2)$ to $f(R^1)$ (he even \sd-prefers $f(R^2)$ to $f(R^1)$). Hence, $f$ is \pc-manipulable if $f(R^1,c)>0$, contradicting our assumptions. 
		
		As the second case, assume that $f(R^1,b)=f(R^2,c)>0$ (note that this is \emph{not} symmetric to the case studied in the previous paragraph) and consider the following sequence of preference profiles $\bar R^{i}$ for $i\in \{0, \dots, \frac{n-1}{2}\}$.
		
		\begin{profilenarrow}
	    $\bar R^{i}$: & $[1\dots \frac{n-1}{2}]$: $b,a,c$ & $\frac{n+1}{2}$: $a,b,c$
	    & $[\frac{n+3}{2}\dots n-i]$: $c,a,b$ & $[n-i+1\dots n]$ $a,c,b$ 
		\end{profilenarrow}
		
		First, note that $R^1=\bar R^{0}$ and  \pc-efficiency shows for all profiles $\bar R^{i}$ that $f(\bar R^{i},b)=0$ or $f(\bar R^{i},c)=0$. Moreover, \pc-strategyproofness and \pc-efficiency imply that if $f(\bar R^{i},b)>0$, then $f(\bar R^{i+1})=f(\bar R^{i})$. The reason for this is that if $f(\bar R^{i},b)>0$, then $f(\bar R^{i},c)=0$ because of \pc-efficiency. This means that every lottery with $f(\bar R^{i+1}, b)=0$ is a \pc-manipulation for the deviating voter $n-i$ as he even \sd-prefers $f(\bar R^{i+1})$ to $f(\bar R^{i})$. Hence, $f(\bar R^{i+1},b)>0$, and we can now use \pc-efficiency to derive that $f(\bar R^{i+1},c)=f(\bar R^{i},c)=0$. Finally, \Cref{lem:2alts} implies that $f(\bar R^{i+1})=f(\bar R^{i})$. As a consequence, this sequence ends at a profile $R^3=\bar R^{\frac{n-1}{2}}$ with $f(R^3)=f(R^1)$.
		
		Next, consider the profile $R^4$ which is derived from $R^3$ by swapping $b$ and $c$ in the preference relation of voter $\frac{n+1}{2}$.
		
		\begin{profile}
	    $R^4$: & $[1\dots \frac{n-1}{2}]$: $b,a,c$ & $\frac{n+1}{2}$: $a,c,b$ 
	    & $[\frac{n+3}{2}\dots n]$: $a,c,b$ 
		\end{profile}
		
		Since $a$ Pareto-dominates $c$ in $R^4$, it follows that $f(R^4,c)=0$. Hence, we can use again \Cref{lem:2alts} to conclude that $f(R^4)=f(R^3)=f(R^1)$.
		
		As the last step, consider the sequence of profiles $\hat R^{i}$ for $i\in \{0,\dots, \frac{n-1}{2}\}$, which leads from $R^4$ to $R^2$.
		
		\begin{profilenarrow}
	    $\hat R^{i}$: & $[1\dots \frac{n-1}{2}]$: $b,a,c$ & $\frac{n+1}{2}$: $a,c,b$ 
	    & $[\frac{n+3}{2}\dots n-i]$: $a,c,b$ & $[n-i+1\dots n]$ $c,a,b$ 
		\end{profilenarrow}
		
		First, observe that $\hat R^{0}=R^4$ and $\hat R^{\frac{n-1}{2}}=R^2$. Moreover, \pc-efficiency requires again for every profile $\hat R^{i}$ that either $f(\hat R^{i},b)=0$ or $f(\hat R^{i},c)=0$. Even more, since $f(\hat R^{0}, c)=0$ and $f(\hat R^{\frac{n-1}{2}},b)=0$, there is at least one index $i$ such that $f(\hat R^{i},c)=f(\hat R^{i+1},b)=0$. Let $i^*\in\{0,\dots,\frac{n-3}{2}\}$ denote the smallest such index, which means that $f(\hat R^{i},c)=0$ for all $i\in \{0,\dots,i^*\}$. Therefore, we can again use \Cref{lem:2alts} to conclude that $f(\hat R^{i^*})=f(\hat R^{0})=f(R^1)$, which means in particular that $f(\hat R^{i^*},b)=f(R^1,b)>0$. Now, if $f(\hat R^{i^*+1},a)\geq f(R^{i^*},a)$, voter $n-{i^*}$ can \pc-manipulate by deviating from $\hat R^{i^*}$ to $\hat R^{i^*+1}$. This follows as voter $n-i^*$, whose preference is $a,c,b$ in $R^{i^*}$, \sd-prefers (and therefore also \pc-prefers) $f(\hat R^{i^*+1})$ to $f(\hat R^{i^*})$ in this case. Hence, \pc-strategyproofness requires that $f(\hat R^{i^*+1},a)< f(\hat R^{i^*},a)$. Since $f(\hat R^{i^*+1},b)=f(\hat R^{i^*},c)=0$, this implies that $f(\hat R^{i^*+1},c)>f(\hat R^{i^*},b) = f(R^1,b)$.
		
		Next, we prove that $f(\hat R^{i+1},c)\geq f(\hat R^{i},c)$ for all $i>i^*$. Assume for contradiction that there is an index $j$ where this is not the case. Then, there is also a minimal index $j^*>i^*$ such that $f(\hat R^{j^*+1},c)< f(\hat R^{j^*},c)$. 
		In particular, it follows from the minimality of $j^*$ that $f(\hat R^{j^*},c)\geq f(\hat R^{i^*+1},c)>0$ and \pc-efficiency then shows that $f(\hat R^{j^*},b)=0$. Now, note that voter $n-j^*$'s preference relation in $\hat R^{j^*+1}$ is $c,a,b$. Hence, if $f(R^{j^*+1},c)=0$, he clearly \pc-prefers $f(\hat R^{j^*})$ to $f(\hat R^{j^*+1})$. This means that $f(\hat R^{j^*+1},c)>0$ and consequently $f(\hat R^{j^*+1},b)=0$ because of \pc-efficiency. However, then voter $j^*$ still \pc-prefers $f(\hat R^{j^*})$ to $f(\hat R^{j^*+1})$ because $f(\hat R^{j^*+1},c)<f(\hat R^{j^*},c)$. Hence, voter $j^*$ can either way \pc-manipulate by deviating from $\hat R^{j^*+1}$ to $\hat R^{j^*}$. This contradicts the \pc-strategyproofness of $f$, and so we must have $f(\hat R^{i+1},c)\geq f(\hat R^{i},c)$ for all $i>i^*$. In particular, this implies that $f(R^2,c)\geq f(\hat R^{i^*+1},c)>f(\hat R^{i^*},b)=f(R^1,b)$ because $R^2=\hat R^{\frac{n-1}{2}}$. However, this observation is in conflict with anonymity and neutrality between $R^2$ and $R^1$, and thus, the assumption that $f(R^1,b)>0$ must be wrong. It follows that $f(R^1,b)=f(R^1,c)=0$, and so $f(R^1,a)=1$. 
		
		Finally, departing from the insight that $f(R^1,a)=1$, we can essentially apply the same steps as in the proof for even $n$ (in the main body) to show that $f$ must satisfy the absolute winner property.
	\end{proof}

\subsection[Proofs of Propositions 3 and 4]{Proofs of \Cref{prop:m3sp,prop:example-m-3}}\label{subsec:propositions}

In this subsection, we prove \Cref{prop:m3sp,prop:example-m-3}, which show that the impossibilities for $m \ge 4$ in \Cref{thm:CCSP,thm:EffSP,thm:strictpart} turn into possibilities when $m = 3$. We first provide additional insights into \pc-efficiency that facilitate the analysis of $f^1$ and $f^2$. 

\begin{lemma}\label{lem:eff3}
Consider a profile $R\in \mathcal{R}^{\mathcal{F}(\mathbb{N})}$ on three alternatives $A=\{a,b,c\}$. A lottery~$p$ is $\pc$-efficient for $R$ if it satisfies the following conditions.
\begin{enumerate}
    \item $p(x)=0$ if $x$ is Pareto-dominated in $R$.
    \item For an alternative $x\in A$ that is never bottom-ranked and at least once top-ranked in $R$, there is $y\in A\setminus \{x\}$ with $p(y)=0$.
    \item For an alternative $x\in A$ that is never top-ranked and at least once bottom-ranked in $R$, $p(x)=0$.
\end{enumerate}
\end{lemma}
\begin{proof}
Consider an arbitrary electorate $N\in \mathcal{F}(\mathbb{N})$ and let $R\in \mathcal{R}^N$ denote a profile. Moreover, let $p$ denote a lottery that satisfies the given conditions and suppose for contradiction that there is another lottery $q$ that \pc-dominates $p$ on $R$. Hence, $p\neq q$, which means that there are alternatives $x,y$ such that $q(x)>p(x)$ and $q(y)<p(y)$. We suppose subsequently that $q(a)>p(a)$ and $q(b)<p(b)$ because our arguments are completely symmetric. Next, we proceed with a case distinction with respect to the relation between $q(c)$ and $p(c)$.\medskip

\textbf{Case 1: $q(c)=p(c)$}

As the first case, we suppose that $q(c)=p(c)$. Then, it follows for all voters $i\in N$ that $a\succ_i b$ because $q \succsim_i^\pc p$. Indeed, if $b\succ_i a$ for some $i\in N$, this voter strictly \sd-prefers $p$ to $q$. Since $q$ \pc-dominates $p$ by assumption, it thus follows that $a$ Pareto-dominates $b$. However, condition~1 then requires that $p(b)=0$, which contradicts that $q(b)<p(b)$. Hence, $q$ cannot \pc-dominate $p$ in this case.\medskip

\textbf{Case 2: $q(c)<p(c)$}

Next, suppose that $q(c)<p(c)$. Combined with $q(a)>p(a)$, $q(b)<p(b)$, and the assumption that $q$ \pc-dominates $p$, this means that no voter bottom-ranks $a$. Indeed, it is easy to see that such a voter strictly \sd-prefers $p$ to $q$, contradicting the \pc-dominance of $q$. Now, if $a$ is top-ranked by a voter in $R$, then condition~2 requires that either $p(c)=0$ or $p(b)=0$. However, this is not possible since $q(c)<p(c)$ and $q(b)<p(b)$ by assumption. Hence, every voter ranks $a$ at the second position in $R$. 

Furthermore, both $b$ and $c$ must be top-ranked at least once; otherwise, one of these alternatives is unanimously top-ranked and therefore Pareto-dominates both other alternatives, which again conflicts with $q(c)<p(c)$ and $q(b) < p(b)$. 

Hence, every voter in $R$ has the preference relation $b,a,c$ or $c,a,b$, and each of these two preferences is submitted at least once. Voters of the first type \pc-prefer $q$ to $p$ if $q(b)p(a)+q(b)p(c)+q(a)p(c)\geq p(b)q(a)+p(b)q(c)+p(a)q(c)$ and voters of the second type if $ q(c)p(a)+q(c)p(b)+q(a)p(b)\geq p(c)q(a)+p(c)q(b)+p(a)q(b)$. Clearly, both inequalities are only true if they hold with equality. However, then no voter strictly \pc-prefers $q$ to $p$ and hence $q$ does not \pc-dominate $p$ in this case either.\medskip

\textbf{Case 3: $q(c)>p(c)$}

As the last case, we assume that $q(c)>p(c)$. Since $q(a)>p(a)$, $q(b)<p(b)$, and $q$ \pc-dominates $p$ in $R$, no voter top-ranks $b$ in $R$. Indeed, such a voter strictly \sd-prefers $p$ to $q$, which contradicts that $q$ \pc-dominates $p$. Now, if $b$ is bottom-ranked by at least one voter in $R$, then condition~3 requires that $p(b)=0$, which contradicts $q(b)<p(b)$. Hence, $b$ is second-ranked by all voters in $R$. Next, both $a$ and $c$ are top-ranked at least once in $R$; otherwise, $b$ is Pareto-dominated which again contradicts $q(b)<p(b)$. Hence, every voter in $R$ has the preference $a,b,c$ or $c,b,a$, and both preferences are reported at least once. An analogous argument as in Case 2 implies that $q$ cannot \pc-dominate $p$ in this case either.
\end{proof}

Next, we use \Cref{lem:eff3} to prove \Cref{prop:m3sp}. Recall the definition of $f^1$ (where $\mathit{CW}(R)$ is the set of Condorcet winners in $R$ and $\mathit{WCW}(R)$ the set of weak Condorcet winners).
\[f^1(R)=\begin{cases}
[x:1] &\text{if } \textit{CW}(R)=\{x\} \\
[x:\frac{1}{2}; y:\frac{1}{2}]&\text{if } \textit{WCW}(R)=\{x,y\} \\
[x:\frac{3}{5}; y:\frac{1}{5}; z:\frac{1}{5}] &\text{if } \textit{WCW}(R)=\{x\} \\
[x:\frac{1}{3}; y:\frac{1}{3}; z:\frac{1}{3}] &\text{otherwise}
\end{cases}
\]

\possp*
\begin{proof}
    The proposition consists of two claims: on the one hand, we need to show that $f^1$ satisfies all axioms of the proposition, and on the other hand, that $f^1$ is the only SDS satisfying these axioms. We consider both claims separately and start by showing that $f^1$ satisfies all axioms of the proposition.\medskip
    
    \textbf{Claim 1: $f^1$ satisfies anonymity, neutrality, cancellation, \pc-efficiency and \pc-strategyproofness.}
    
    First, note that $f^1$ satisfies cancellation because adding two voters with inverse preferences does not affect whether an alternative is a (weak) Condorcet winner. Furthermore, the definition of $f^1$ immediately shows that it is anonymous and neutral. 
    
    For proving that $f^1$ is \pc-efficient, we consider an arbitrary preference profile $R\in\mathcal{R}^{\mathcal{F}(\mathbb{N})}$. Now, if an alternative $x$ is Pareto-dominated in $R$, then it is never top-ranked. Consequently, there is either a Condorcet winner $y\neq x$ (if more than half of the voters top-rank $y$) or the remaining two alternatives $y,z$ are weak Condorcet winners (if both $y$ and $z$ are top-ranked by exactly half of the voters). In both cases, $f^1(R,x)=0$, 
    which shows that $f^1(R)$ satisfies condition~1 of \Cref{lem:eff3} for all profiles $R$. Similarly, if there is an alternative $x$ that is never top-ranked and at least once bottom-ranked, then either there is a Condorcet winner $y\neq x$, or the remaining two alternatives $y,z$ are weak Condorcet winners. Hence, $f^1(R,x)=0$, which proves that $f^1(R)$ also satisfies condition~3 of \Cref{lem:eff3}. Finally, if there is an alternative $x$ that is never bottom-ranked and at least once top-ranked in $R$, then there is an alternative $z\neq x$ with $g_R(x,z)>0$. We claim that $f^1(R,z)=0$, which proves condition~2 of \Cref{lem:eff3}. If $x$ or the third alternative $y$ is a Condorcet winner, this follows immediately. On the other hand, if neither $x$ nor $y$ are Condorcet winners, both of them are weak Condorcet winners because $y$ must be top-ranked by at least half of the voters if $x$ is not a Condorcet winner. Hence, we have two weak Condorcet winners and the definition of $f^1$ again shows that $f^1(R,z)=0$. Since all conditions of \Cref{lem:eff3} hold, it thus follows that $f^1(R)$ is \pc-efficient for all profiles $R$. 
    
    Finally, we need to show that $f^1$ is \pc-strategyproof.
    Assume for contradiction that this is not the case. Then, there are an electorate $N$, two preference profiles $R,R'\in\mathcal{R}^N$, and a voter $i\in N$ such that $f^1(R')\succ_i^\pc f^1(R)$ in $R$, and $R_{-i}=R_{-i}'$. Subsequently, we discuss a case distinction with respect to the definition to $f^1$. In more detail, we have for both $R$ and $R'$ five different options: there is a Condorcet winner ($CW$), or there is no Condorcet winner but $k\in \{0,1,2,3\}$ weak Condorcet winners ($kWCW$). We label the cases with a shorthand notation: for instance, $CW\rightarrow 1WCW$ is the case where there is a Condorcet winner in $R$ and a single weak Condorcet winner in $R'$. 
    
    To keep the length of the proof manageable, we subsequently focus only on the case that $R$ and $R'$ are defined by an odd number of voters. This assumption means that there are no weak Condorcet winners and thus significantly reduces the number of cases that need to be considered. For the case that $R$ and $R'$ are defined by an even number of voters, we refer to pages 20--22 of a preprint of this paper \citep{BLS22b}. When the number of voters $n$ is odd, there are only four possible types of manipulations. 
    \begin{itemize}
        \item $CW\rightarrow CW$: Suppose that $a$ is the Condorcet winner in $R$. If $a$ is also the Condorcet winner in $R'$, then $f^1(R)=f^1(R')$ and deviating from $R$ to $R'$ is no \pc-manipulation. On the other hand, if another alternative $b$ is the Condorcet winner in $R'$, we must have $a\succ_i b$ in $R$. Since $f^1(R,a)=f^1(R',b)=1$, this is no \pc-manipulation.
        \item $CW\rightarrow 0WCW$: Suppose that $a$ is the Condorcet winner in $R$, and there is no Condorcet winner in $R'$. This means that voter $i$ reinforces an alternative $b$ against~$a$, i.e., $a$ is ranked either second or third. Since $f^1(R,a)=1$ and $f^1(R',x)=\frac{1}{3}$ for all $x\in A$, this proves that deviating from $R$ to $R'$ is no \pc-manipulation.
        \item $0WCW\rightarrow CW$: Suppose there is no Condorcet winner in $R$, but $a$ is the Condorcet winner in $R'$. Hence, voter $i$ needs to reinforce $a$ against at least one other alternative $b$. This means that $a$ is not voter $i$'s favorite alternative in~$R$. Since $f^1(R,x)=\frac{1}{3}$ for all $x\in A$ and $f(R',a)=1$, this observation proves that $f^1$ is \pc-strategyproof in this case.
        \item $0WCW\rightarrow 0WCW$: We have $f^1(R)=f^1(R')$ in this case, which contradicts that voter $i$ can \pc-manipulate. 
    \end{itemize}\medskip
    
    \textbf{Claim 2: $f^1$ is the only SDS that satisfies anonymity, neutrality, cancellation, \pc-efficiency, and \pc-strategy\-proofness.}
    
    Consider an arbitrary SDS $f$ for $m=3$ alternatives that satisfies all given axioms. We show that $f(R)=f^1(R)$ for all profiles $R\in\mathcal{R}^{\mathcal{F}(\mathbb{N})}$, which proves this claim. For this, we name the six possible preference relations ${\succ_1}=a,b,c$, ${\succ_2}=c,b,a$, ${\succ_3}=b,c,a$, ${\succ_4}=a,c,b$, ${\succ_5}=c,a,b$, and ${\succ_6}=b,a,c$. Moreover, given a profile $R$, let $n_i$ denote the number of voters who report preference relation $\succ_i$ in $R$. Using this notation, we can describe the majority margins of $R$ as follows. 
    \begin{align*}
        &g_R(a,b)=(n_1-n_2)-(n_3-n_4)+(n_5-n_6) \\
        &g_R(b,c)=(n_1-n_2)+(n_3-n_4)-(n_5-n_6) \\
        &g_R(c,a)=-(n_1-n_2)+(n_3-n_4)+(n_5-n_6)
    \end{align*}
     It is not difficult to derive from these equations that
     \begin{align*}
         &n_1=\frac{g_R(a,b)+g_R(b,c)}{2}+n_2 \\
         &n_3=\frac{g_R(b,c)+g_R(c,a)}{2}+n_4 \\
         &n_5=\frac{g_R(c,a)+g_R(a,b)}{2}+n_6.
     \end{align*}
    
    Next, consider an arbitrary preference profile $R$. Based on cancellation, we can use the above equations to remove pairs of voters with inverse preferences from $R$ until $n_{2k}=0$ or $n_{2k-1}=0$ for each $k\in \{1,2,3\}$. Unless all majority margins are $0$, this leads to a minimal profile $R'$, which we consider in the subsequent case distinction. Note that the removal of voters with inverse preferences does not affect the majority margins and therefore also not the (weak) Condorcet winners. In particular, this means that $f^1(R)=f^1(R')$. Analogously, cancellation yields for $f$ that $f(R)=f(R')$. Hence, we will consider multiple cases depending on the structure of $R'$ and prove that $f(R)=f(R')=f^1(R')=f^1(R)$ in every case. On the other hand, if all majority margins are $0$, we need a separate argument, which we discuss in our first case below. Taken together, our cases imply that $f(R)=f^1(R)$ for every profile $R$.\smallskip

    \textit{Case 2.1: $g_R(a,b)=g_R(b,c)=g_R(c,a)=0$.}
    
    First, suppose that $g_R(a,b)=g_R(b,c)=g_R(c,a)=0$, which means that all three alternatives are weak Condorcet winners in $R$. Our equations show that $n_1=n_2$, $n_3=n_4$, and $n_5=n_6$. Let $n^*$ denote the maximum among all $n_i$. Using cancellation, we can add pairs of voters with inverse preferences until $n_k=n^*$ for every $k\in \{1,\dots, 6\}$. Moreover, cancellation implies that $f(R)=f(R'')$ for the new profile $R''$. Finally, all alternatives are symmetric to each other in $R''$ since all preference relations appear equally often. Hence, anonymity and neutrality require that $f(R'',x)=\frac{1}{3}$ for all $x\in A$, which means that $f(R)=f(R'')=f^1(R)$.\smallskip
    
    \textit{Case 2.2: An alternative $x$ is top-ranked by more than half of the voters in $R'$.}
    
    As the second case, suppose that $R'$ is well-defined and that an alternative $x$ is top-ranked by more than half of the voters in this profile. Then, it holds that $f(R',x)=1$ because \Cref{lem:abswinner} implies that $f$ satisfies the absolute winner property. Since $x$ is the Condorcet winner in $R'$, it holds that $f(R')=f^1(R')$.\smallskip
    
    \textit{Case 2.3: Two alternatives are top-ranked by exactly half of the voters in $R'$.}
    
    Next, suppose that $R'$ is well-defined and that two alternatives, say $a$ and $b$, are top-ranked by exactly half of the voters in $R'$. Then, $a$ and $b$ are weak Condorcet winners. Moreover, $c$ is not a weak Condorcet winner in $R'$ since not all majority margins in $R$ can be $0$, which implies that there is a voter who ranks $c$ last. Due to symmetry, we can assume that this voter's preference relation is $a,b,c$. Now, if there is a voter with preference relation $a,c,b$ in $R'$, then the last possible preference relation is $b,a,c$; otherwise, $R'$ is not minimal. Hence, $a$ Pareto-dominates $c$ in $R'$. Similarly, if there is no voter with the preference $a,c,b$, all voters prefers $b$ to $c$ and $c$ is again Pareto-dominated. Therefore, it follows in both cases that $f(R',c)=0$ because of \pc-efficiency. Moreover, we can let the voters with $a,c,b$ and $b,c,a$ (if any) push down $c$. Then, $c$ stays Pareto-dominated and therefore still receives probability $0$ from $f$. Hence, \Cref{lem:2alts} shows that the probability of $a$ and $b$ does not change during these steps. Finally, this process results in a profile $R''$ in which half of the voters report $a,b,c$ and the other half $b,a,c$. Anonymity, neutrality, and \pc-efficiency imply for this profile $R''$ that $f(R'',a)=f(R'',b)=\frac{1}{2}$. Hence, we have that $f(R')=f(R'')=f^1(R')$ because $a$ and $b$ are the only weak Condorcet winners in $R'$.\smallskip
    
    \textit{Case 2.4: Each alternative is top-ranked at least once and one alternative is top-ranked by exactly half of the voters in $R'$.}
    
    Next, suppose that an alternative is top-ranked by exactly half of the voters and the other two alternatives are top-ranked at least once. Without loss of generality, assume that there is a voter with preference relation $a,b,c$ in $R'$. Since $c$ is top-ranked by a voter, there is also a voter with preference relation $c,a,b$; note for this that no voter can report $c,b,a$ in $R'$ because of the minimality of $R'$. By an analogous argument, we also derive that there is a voter with preference relation $b,c,a$. In summary, we have that $n_1>0$, $n_3>0$, $n_5>0$, and $n_2=n_4=n_6=0$. Moreover, one alternative is top-ranked by half of the voters; suppose without loss of generality that this alternative is $a$. Hence, $n_1=n_3+n_5$. 
    We prove that $f(R',a)=\frac{3}{5}$ and $f(R',b)=f(R',c)=\frac{1}{5}$ by considering the following preference profiles, where $l=n_1+n_3$.
    
     \begin{profilenarrow}
	   $R^{1,n_3,n_5}$: & $[1\dots n_1]$: $a,b,c$ & $[n_1+1\dots l]$: $b,c,a$
	    & $[l+1\dots n-1]$: $c,a,b$  & $n$: $c,b,a$\\
	   $R^{2,n_3,n_5}$: & $[1\dots n_1]$: $a,b,c$ & $[n_1+1\dots l]$: $b,c,a$ 
	   & $[l+1\dots n]$: $c,a,b$\\
	   $R^{3,n_3,n_5}$: & $[1\dots n_1]$: $a,b,c$ & $[n_1+1\dots l-1]$: $b,c,a$ 
	    & $l$: $c,b,a$ & $[l+1\dots n]$: $c,a,b$
	\end{profilenarrow}
	
	Anonymity implies that $f(R')=f(R^{2,n_3,n_5})$. 
	Hence, our goal is to show that $f(R^{2,n_3,n_5},a)=\frac{3}{5}$ and $f(R^{2,n_3,n_5},b)=f(R^{2,n_3,n_5},c)=\frac{1}{5}$ for all $n_3>0$ and $n_5>0$. Note for this that, in $R^{1,n_3,n_5}$ and $R^{3,n_3,n_5}$, we can use cancellation to remove voters $1$ and $n$ or voters $1$ and $l$, respectively. This step leads to the profile $R^{2,n_3,n_5-1}$ or $R^{2, n_3-1, n_5}$, which proves that $f(R^{1,n_3,n_5})=f(R^{2,n_3,n_5-1})$ and $f(R^{3,n_3,n_5})=f(R^{2,n_3-1,n_5})$. Moreover, note that if $n_3=0$, then $a$ and $c$ are top-ranked by half of the voters in $R^{2,n_3,n_5}$. Hence, we have that $f(R^{2,0,n_5},a)=f(R^{2,0,n_5},c)=f(R^{3,1,n_5},a)=f(R^{3,1,n_5},c)=\frac{1}{2}$ by Case 2.3. An analogous argument also shows that $f(R^{2,n_3,0},a)=f(R^{2,n_3,0},b)=f(R^{1,n_3,1},a)=f(R^{1,n_3,1},b)=\frac{1}{2}$. Based on these insights, we now prove our claim on $f(R^{2,n_3,n_5})$ with an induction on $n_3+n_5$.
	
	First, we consider the induction basis that $n_3=n_5=1$. The previous paragraph implies that $f(R^{1,n_3,n_5},a)=f(R^{1,n_3,n_5},b)=f(R^{3,n_3,n_5},a)=f(R^{3,n_3,n_5},c)=\frac{1}{2}$. Hence, \pc-strategyproofness from $R^{1,n_3,n_5}$ to $R^{2,n_3,n_5}$ and from $R^{2,n_3,n_5}$ to $R^{3,n_3,n_5}$ entails the following inequalities, where $p=f(R^{2,n_3,n_5})$. 
	\begin{align*}
	    \frac{1}{2}p(a)\geq p(c)+\frac{1}{2}p(b)\qquad\qquad\qquad p(b)+\frac{1}{2}p(c)\geq \frac{1}{2}p(a)
	\end{align*}
	
	Moreover, note that voter $n$ can ensure in $R^{2,n_3,n_5}$ that $a$ is chosen with probability~$1$ by reporting it as his favorite alternative because of the absolute winner property. Hence, we also get that $p(c)\geq p(b)$ from $\pc$-strategyproofness. Finally, it is easy to see that these three inequalities are true at the same time only if $p(a)=3p(b)=3p(c)$. Using the fact that $p(a)+p(b)+p(c)=1$, we hence derive that $p(a)=\frac{3}{5}$ and $p(b)=p(c)=\frac{1}{5}$. 
	
	Next, we prove the induction step and thus consider some fixed $n_3>0$ and $n_5>0$ such that $n_3+n_5>2$. The induction hypothesis is that $f(R^{2,n_3',n_5'},a)=\frac{3}{5}$ and $f(R^{2,n_3',n_5'},b)=f(R^{2,n_3',n_5'},c)=\frac{1}{5}$ for all $n_3'>0$ and $n_5'>0$ with $n_3'+n_5'=n_3+n_5-1$. Now, recall that $f(R^{1,n_3,n_5})=f(R^{2,n_3,n_5-1})$, which means that $f(R^{1,n_3,n_5},a)=\frac{3}{5}$ and $f(R^{1,n_3,n_5},b)=f(R^{1,n_3,n_5},c)=\frac{1}{5}$ if $n_5>1$ because of the induction hypothesis. \pc-strategyproofness from $R^{1, n_3, n_5}$ to $R^{2, n_3, n_5}$ implies then the following inequality, where $p=f(R^{2,n_3,n_5})$. 
	\begin{align*}
	    \frac{2}{5}p(a)+\frac{1}{5}p(b)\geq \frac{4}{5}p(c)+\frac{3}{5}p(b) \qquad \iff \qquad \frac{1}{2}p(a)\geq p(c)+\frac{1}{2}p(b)
	\end{align*}
	On the other hand, if $n_5=1$, then $f(R^{1,n_3,n_5},a)=f(R^{1,n_3,n_5},b)=\frac{1}{2}$, and \pc-strategyproofness results in the same inequality. 
	
	Similarly, if $n_3>1$, then $f(R^{3,n_3,n_5},a)=\frac{3}{5}$ and $f(R^{3,n_3,n_5},b)=f(R^{3,n_3,n_5},c)=\frac{1}{5}$ because of the induction hypothesis and cancellation. Hence, we derive the following inequality from \pc-strategyproofness between $R^{2,n_3,n_5}$ and $R^{3,n_3,n_5}$. 
	\begin{align*}
	    \frac{4}{5}p(b)+\frac{3}{5}p(c) \geq \frac{1}{5}p(c)+\frac{2}{5}p(a)
	    \qquad \iff\qquad p(b)+\frac{1}{2}p(c) \geq \frac{1}{2}p(a)
	\end{align*}
	On the other hand, if $n_3=1$, then $f(R^{3,n_3,n_5},a)=f(R^{3,n_3,n_5},c)=\frac{1}{2}$. Applying \pc-strategyproofness in this case results in the same inequality as above. 
	
	Finally, it must hold that $p(c)\geq p(b)$.
	Indeed, otherwise voter $n$ could \pc-manipulate in $R^{2,n_3,n_5}$ by reporting $a$ as his favorite option---$a$ would then chosen with probability~$1$ because of the absolute winner property. Since $p(a)+p(b)+p(c)=1$, it can be verified that the only possible solution to the three inequalities that we have derived is $p(a)=\frac{3}{5}$ and $p(b)=p(c)=\frac{1}{5}$. This proves the induction step and therefore that $f(R')=f(R^{2,n_3,n_5})=f^1(R')$.\smallskip
	
	\textit{Case 2.5: Every alternative is top-ranked by less than half of the voters in $R'$.}
    
    As the last case, suppose that every alternative is top-ranked by less than half of the voters in $R'$. In particular, this means that every alternative is top-ranked at least once. We suppose again without loss of generality that a voter reports $a,b,c$ in $R'$ and hence, the same analysis as in the previous case shows that the only possible preference relations in $R'$ are $\succ_1=a,b,c$, $\succ_3=b,c,a$, and $\succ_5=c,a,b$. In particular, we have that $n_1>0$, $n_3>0$, $n_5>0$, and $n_2=n_4=n_6=0$. Moreover, since no alternative is top-ranked by at least half of the voters, we have that $n_1<n_3+n_5$, $n_3<n_1+n_5$, and $n_5<n_1+n_3$. This shows that there is not even a weak Condorcet winner in $R'$, and our goal hence is to show that $f(R',x)=\frac{1}{3}$ for all $x\in A$. Suppose that this is not the case, which means that that either $f(R',a)<f(R',c)$, $f(R',b)<f(R',a)$, or $f(R',c)<f(R',b)$; otherwise, $f(R',a)\geq f(R',c)\geq f(R',b)\geq f(R',a)$, which implies that all alternatives get a probability of $\frac{1}{3}$. We assume in the sequel that $f(R',a)<f(R',c)$ as all cases are symmetric. Now, in this case, we let the voters $i$ with preference relation $a,b,c$ one after another swap $a$ and $b$. For each step, \Cref{lem:3alts} implies that the probability of $a$ remains smaller than that of $c$. However, this process results in a profile $R''$ in which $n_1+n_3$ voters report $b$ as their favorite alternative. Since $n_1+n_3>n_5$, $b$ is the absolute winner and Case~2.2 shows that $f(R'',b)=1$. However, this contradicts $f(R'',a)<f(R'',c)$ and hence, the claim that $f(R',a)<f(R',c)$ must be wrong. This proves that $f(R',x)=\frac{1}{3}=f^1(R',x)$ for all $x\in A$. 
\end{proof}

Finally, we prove \Cref{prop:example-m-3}. Recall for this that $n_R(x)$ denotes the number of voters who top-rank alternative $x$ in $R$, and let $B(R)$ be the set of alternatives that are never bottom-ranked in $R$. Moreover, the uniform random dictatorship $\rd$ is defined by $\rd(R,x)=\frac{n_R(x)}{\sum_{y\in A}n_R(y)}$ for all $x\in A$ and $R\in\mathcal{R}^{\mathcal{F}(\mathbb{N})}$. As discussed in \Cref{sec:RD-ML}, \rd is known to satisfy strict \sd-participation and therefore satisfies also strict \pc-participation, but fails \pc-efficiency. We consider the following variant of \rd called $f^2$: if $|B(R)|\in\{0,2\}$, then $f^2(R)=\rd(R)$.
On the other hand, if $|B(R)|=1$, let $x$ denote the single alternative in $B(R)$ and let $C$ denote the set of alternatives $y\in A\setminus \{x\}$ with minimal $n_y(R)$. Then, $f^2(R,x)=\frac{n_R(x)+\sum_{y\in C} n_R(y)}{\sum_{y\in A} n_R(y)}$, $f(R^2,y)=0$ for $y\in C$, and $f^2(R,z)=\rd(R,z)$ for $z\not\in C\cup \{x\}$. Intuitively, if $|B(R)| = 1$, $f^2$ removes the alternatives in $A\setminus B(R)$ with minimal $n_R(x)$ and then computes $\rd$.

\pospart*
\begin{proof}
First note that $f^2$ is anonymous and neutral since its definition does not depend on the identities of voters or alternatives. 

Next, we show that $f^2$ satisfies \pc-efficiency by proving that $f^2(R)$ satisfies for all profiles $R$ the three conditions of \Cref{lem:eff3}. To this end, note first that $f^2$ is \emph{ex post} efficient: it only puts positive probability on an alternative that is never top-ranked if it is second-ranked by all voters and both other alternatives are top-ranked at least once. In this case, all three alternatives are Pareto-optimal, and thus $f^2$ is \emph{ex post} efficient. This argument also shows that an alternative that is never top-ranked and at least once bottom-ranked is always assigned probability $0$. Finally, if an alternative is never bottom-ranked and at least once top-ranked, only two alternatives can have positive probability. In more detail, either $|B(R)|=2$, which means that one alternative is bottom-ranked by all voters and receives probability $0$, or $|B(R)|=1$ and an alternative in $A\setminus B(R)$ gets probability $0$ by definition of $f^2$. Hence, all conditions of \Cref{lem:eff3} hold, which implies that $f^2$ is \pc-efficient.

Lastly, we discuss why $f^2$ satisfies strict \pc-participation---in fact, we prove the even stronger claim that it satisfies strict \sd-participation. Consider an arbitrary electorate $N\in\mathcal{F}(\mathbb{N})$, a voter $i\in N$, and two preference profiles $R\in\mathcal{R}^N$ and $R'\in\mathcal{R}^{N\setminus \{i\}}$ such that $R'=R_{-i}$. We need to show that if $i$'s top alternative is not already chosen with probability $1$ in $f^2(R')$, then $f^2(R)\succ_i^\sd f^2(R')$. First, note that this is obvious if $f^2(R)=\rd(R)$ and $f^2(R')=\rd(R')$ because $\rd$ satisfies strict \sd-participation. Moreover, $|B(R')|-1\leq |B(R)|\leq |B(R')|$ because voter $i$ can only bottom-rank a single alternative. These two observations leave us with three interesting cases: $|B(R')|=2$ and $|B(R)|=1$, $|B(R')|=|B(R)|=1$, and $|B(R')|=1$ and $|B(R)|=0$.

First, consider the case where $|B(R')|=1$ and $|B(R)|=0$. Without loss of generality, we assume that $B(R')=\{a\}$, which means that $a$ is voter $i$'s least preferred alternative. Moreover, we call voter $i$'s best alternative $z\in\{b,c\}$. The following case distinction proves that $f^2$ satisfies strict \sd-participation under the given assumptions.
\begin{itemize}
    \item If $n_{R'}(b)=n_{R'}(c)$, then $f^2(R',a)=1$ and it is obvious that $f^2(R)\succ_i^\sd f^2(R')$ because $a$ is voter $i$'s least preferred outcome and $f^2(R,z)=\rd(R,z)>0$.
    \item If $n_{R'}(b)>n_{R'}(c)$, we have that $f^2(R',a)=\frac{n_{R'}(a)+n_{R'}(c)}{\sum_{x\in A} n_{R'}(x)}>\frac{n_{R'}(a)}{1+\sum_{x\in A} n_{R'}(x)}=f^2(R,a)$ and $f^2(R',z)\leq \frac{n_{R'}(z)}{\sum_{x\in A} n_{R'}(x)}<\frac{1+n_{R'}(z)}{1+\sum_{x\in A} n_{R'}(x)}=f^2(R,z)$. It is now easy to see that $f^2(R)\succ_i^\sd f^2(R')$.
    \item The case $n_{R'}(b)<n_{R'}(c)$ is symmetric to the previous one.
\end{itemize}

Next, consider the case where $|B(R')|=2$ and $|B(R)|=1$. Without loss of generality, we suppose that $B(R')=\{a,b\}$ and $B(R)=\{a\}$, which means that voter $i$ bottom-ranks~$b$. Moreover, note that all voters in $N\setminus \{i\}$ bottom-rank $c$ as otherwise $B(R')=\{a,b\}$ is not possible. This means that $f^2(R',c)=0$, $n_{R'}(c)=0$, and $n_{R}(c)\leq 1$. We consider again several subcases.
\begin{itemize}
    \item If $n_{R}(b)>n_{R}(c)$, then $f^2(R,c)=0=f^2(R',c)$, $f^2(R,a)\ge\frac{1+n_{R'}(a)}{1+\sum_{x\in A} n_{R'} (x)}>\frac{n_{R'}(a)}{\sum_{x\in A} n_{R'} (x)}=f^2(R',a)$, and thus $f^2(R,b)<f^2(R',b)$. Hence, $f^2(R)\succ_i^\sd f^2(R')$ as $b$ is voter $i$'s worst alternative.
    \item If $n_{R}(c)>n_{R}(b)$, then $f^2(R,b)=0\leq f^2(R',b)$ and $f^2(R,c)>0=f^2(R',c)$. 
    If $i$ top-ranks $c$, we have $f^2(R)\succ_i^\sd f^2(R')$.
    Else, $i$ top-ranks $a$, and we have $f^2(R,a)\ge\frac{1+n_{R'}(a)}{1+\sum_{x\in A} n_{R'} (x)}>\frac{n_{R'}(a)}{\sum_{x\in A} n_{R'} (x)}=f^2(R',a)$, so again $f^2(R)\succ_i^\sd f^2(R')$.
    \item If $n_R(c)=n_R(b)=0$, all voters (including $i$) report $a$ as their best option and thus $f^2(R',a)=f^2(R,a)=1$, which satisfies strict \sd-participation because $f^2(R)$ is voter $i$'s favorite lottery.
    \item If $n_{R}(c)=n_{R}(b)=1$, then voter $i$'s preference relation is $c,a,b$. Moreover, $f^2(R',c)=0 = f^2(R,c)$ and $f^2(R',b)>0=f^2(R,b)$. This proves again that $f^2(R)\succ_i^\sd f^2(R')$
\end{itemize}

As the last case, suppose that $|B(R')|=|B(R)|=1$ and let $a$ denote the alternative in $B(R)=B(R')$. Since $a\in B(R)$, voter $i$ does not bottom-rank $a$. We consider again a case distinction.
\begin{itemize}
    \item First, suppose that voter $i$ top-ranks $a$, which means that $n_R(b)=n_{R'}(b)$ and $n_{R}(c)=n_{R'}(c)$.
    \begin{itemize}
        \item If $n_{R'}(b)=n_{R'}(c)$, we have that $f^2(R,a)=f^2(R',a)=1$ and strict \pc-participation holds as this is voter $i$'s favorite lottery. 
        \item If $n_{R'}(b)>n_{R'}(c)$. Then, $f^2(R,a)=\frac{n_{R'}(a)+n_{R'}(c)+1}{1+\sum_{x\in A} n_{R'}(x)}>\frac{n_{R'}(a)+n_{R'}(c)}{\sum_{x\in A} n_{R'}(x)}=f^2(R',a)$, $f^2(R,c)=f^2(R',c)=0$, and hence $f^2(R,b)<f^2(R',b)$. It is now easy to verify that $f(R)\succ_i^\sd f(R')$. 
        \item The case $n_{R'}(b)<n_{R'}(c)$ is symmetric to the previous one.
    \end{itemize}
    \item Next, suppose that voter $i$ places $a$ second. We assume without loss of generality that ${\succ_i}=b,a,c$ because the case ${\succ_i}=c,a,b$ is symmetric. This assumption means that $n_{R'}(b)+1=n_R(b)$ and $n_{R'}(x)=n_R(x)$ for $x\in \{a,c\}$.
    \begin{itemize}
        \item If $n_{R'}(b)\geq n_{R'}(c)$, then $f^2(R',b)\leq\frac{n_{R'}(b)}{\sum_{x\in A} n_{R'}(x)} <\frac{1+n_{R'}(b)}{1+\sum_{x\in A} n_{R'}(x)}=f^2(R,b)$, $f^2(R',c)=f^2(R,c)=0$, and hence $f^2(R',a)>f^2(R,a)$, which proves that $f(R)\succ_i^\sd f(R')$.
        \item If $n_{R'}(b)+1=n_{R'}(c)$, then $f^2(R',b)=0=f^2(R,b)$, $f^2(R',c)>0=f^2(R,c)$, and thus $f^2(R',a)<1=f^2(R,a)$. It can again be verified that $f(R)\succ_i^\sd f(R')$.
        \item If $n_{R'}(b)+1<n_{R'}(c)$, then $f^2(R',b)=0=f^2(R,b)$, $f^2(R',c)=\frac{n_{R'}(c)}{\sum_{x\in A} n_{R'}(x)}>\frac{n_{R'}(c)}{1+\sum_{x\in A} n_{R'}(x)} = f^2(R,c)$, and hence $f^2(R',a)<f^2(R,a)$. Once again, it holds that $f^2(R)\succ_i^\sd f^2(R')$. \qedhere
    \end{itemize}
\end{itemize}
\end{proof}

\end{document}